\documentclass[conference]{IEEEtran}
\IEEEoverridecommandlockouts
\newcommand{\subparagraph}{}
\usepackage{algorithmicx}
\usepackage{lipsum}
\usepackage{tikz}
\usetikzlibrary{patterns}
\usepackage{algorithm}
\usepackage[noend]{algcompatible}
\usetikzlibrary{decorations.pathmorphing}
\usetikzlibrary{decorations,decorations.markings}
\usepackage{etoolbox}
\usetikzlibrary{optics}
\usepackage{enumerate}
\usepackage{stackengine}
\usetikzlibrary{decorations.text}
\usepackage{amsmath,amsfonts,amsthm,amssymb,mathbbol}
\usepackage{graphicx}
\usepackage{grffile}
\usepackage{stackengine}
\usepackage{caption}
\usepackage{subcaption}
\usepackage{titlesec}
\usepackage{paralist}
\usepackage{makecell}
\usepackage{multirow}
\usepackage{float}
\usepackage{tcolorbox}
\usepackage{booktabs}
\titleformat{\section}
{\normalfont\fontsize{13}{16}\bfseries}{\thesection}{1em}{}
\titleformat{\subsection}
{\normalfont\fontsize{11}{14}\bfseries}{\thesubsection}{1em}{}

\newtheorem{lemma}{Lemma}

\renewcommand{\emptyset}{\varnothing}

\theoremstyle{definition}
\newtheorem{theorem}{Theorem}

\newtheorem*{definition}{Definition}
\newtheorem{example}{Example}

\newboolean{showcomments}
\setboolean{showcomments}{false}

\ifthenelse{\boolean{showcomments}}
{ \newcommand{\mynote}[2]{
    \fbox{\bfseries\sffamily\scriptsize#1}
    {\small$\blacktriangleright$\textsf{\emph{#2}}$\blacktriangleleft$}}}
{ \newcommand{\mynote}[2]{}}

\newcounter{myenumi}
\renewcommand{\themyenumi}{\textit{Class} \arabic{myenumi}}

\newcommand{\fbseries}{\unskip\setBold\aftergroup\unsetBold\aftergroup\ignorespaces}
\makeatletter
\newcommand{\setBoldness}[1]{\def\fake@bold{#1}}
\makeatother

\usepackage{mathtools}

\usepackage[T1]{fontenc}
\AtBeginDocument{%
	\begingroup\lccode`~=`_%
	\lowercase{\endgroup\let~}_%
	\catcode`_=12
}

\algrenewcommand\alglinenumber[1]{\footnotesize #1)}

\usepackage[export]{adjustbox}

\renewcommand{\COMMENT}[2][.5\linewidth]{%
\leavevmode\hfill{$\triangleright$~#2}}
\algnewcommand\algorithmicto{\textbf{to}}
\algnewcommand\RETURN{\State \textbf{return} }

\usepackage[final]{pdfpages}
\usepackage{graphicx}
\usepackage{enumitem}

\newtheoremstyle{assumption-style}
  {5pt}
  {5pt}
  {}
  {1pt}
  {\bfseries \small}
  {}
  { }
  {\thmname{#1}\thmnumber{ #2}\textnormal{\thmnote{ (#3)}}}

\theoremstyle{assumption-style}
\newtheorem{assumption}{Assumption}

\newtheoremstyle{property-style}
  {5pt}
  {5pt}
  {}
  {1pt}
  {\bfseries \small}
  {}
  { }
  {\thmname{#1}\thmnumber{ #2}\textnormal{\thmnote{ (#3)}}}

\theoremstyle{property-style}
\newtheorem{property}{Property}

\usepackage{multicol}
\algdef{SE}[SUBALG]{Indent}{EndIndent}{ }{\algorithmicend\ }%
\algtext*{Indent}
\algtext*{EndIndent}

\usepackage{pythonhighlight}
\usepackage{textcomp}
\usepackage{xcolor}

\def\BibTeX{{\rm B\kern-.05em{\sc i\kern-.025em b}\kern-.08em
    T\kern-.1667em\lower.7ex\hbox{E}\kern-.125emX}}
\begin{document}

\title{Efficient Consensus-Free Weight Reassignment for \\ Atomic Storage (Extended Version)}

\author{\IEEEauthorblockN{Hasan Heydari}
\IEEEauthorblockA{\textit{ENAC, Université du Toulouse} \\
Toulouse, France\\	
heydari@enac.fr}
\and
\IEEEauthorblockN{Guthemberg Silvestre}
\IEEEauthorblockA{\textit{ENAC, Université du Toulouse} \\
Toulouse, France\\
silvestre@enac.fr}
\and
\IEEEauthorblockN{Luciana Arantes}
\IEEEauthorblockA{\textit{Sorbonne University, CNRS, Inria} \\
Paris, France\\
luciana.arantes@lip6.fr}
}

\maketitle

\begin{abstract}
	Weighted voting is a conventional approach to improving the performance of
    replicated systems based on commonly-used majority quorum systems in heterogeneous environments.
In long-lived systems, a weight reassignment protocol is required to reassign weights over time
    in order to accommodate performance variations accordingly.
The weight reassignment protocol should be consensus-free in asynchronous failure-prone
    systems because of the impossibility of solving consensus in such systems.
This paper presents an efficient consensus-free weight reassignment protocol for atomic storage systems
    in heterogeneous, dynamic, and asynchronous message-passing systems.
An experimental evaluation shows that the proposed protocol improves the performance of atomic read/write
    storage implemented by majority quorum systems compared with previous solutions.
\end{abstract}


\section{Introduction}\label{sec:introduction}
	The atomic read/write storage (or simply atomic storage, a.k.a. atomic register \cite{interprocessComI})
		is a fundamental building block for practical distributed storage and file systems (e.g., \cite{fab,gpfs}).
	Atomic storage allows concurrent processes, each possibly running a different algorithm,
		to share data atomically through a variable accessed by read/write (r/w) operations.
	Quorum systems \cite{quorumSystems} are a well-known abstraction for implementing atomic storage \cite{effModConsensus-freeFST}.
	A quorum system is a collection of sets called quorums such that each one is a subset of processes,
		and the intersection property that states every two quorums always intersect should be satisfied.
	By implementing atomic storage using quorum systems, atomicity can be guaranteed using the intersection property \cite{abd}.
	Moreover, it is not required to execute r/w operations in all processes;
		each r/w operation should be executed by all processes of one quorum,
		improving the system's fault tolerance and availability.

	There exist many types of quorum systems such as grids \cite{gridQS, loadCapacityAvaiQS},
		trees \cite{treeQS}, hierarchical \cite{hierarchicalQS}, and the simple majority quorum system (SMQS) \cite{readWriteQS}.
	In the SMQS, every quorum consists of a strict majority of processes.
	Most atomic storages based on quorum systems (e.g., \cite{effModConsensus-freeFST, abd, rambo}) utilize the SMQS due to its simplicity and optimal fault tolerance;
		however, the SMQS can impact both quorum latency\footnote{Quorum latency for
		a request in a quorum system is the time interval between sending
		the request (to a quorum, some quorums, or a subset of processes)
		until receiving the responses from a quorum of processes \cite{readWriteQS, minResponceTimeQS}.} and throughput \cite{readWriteQS}.
	The reason for this performance impact is that an SMQS does not consider the heterogeneity of processes or network connections.
	If it takes such heterogeneity into account, its latency and throughput are likely to be improved.

	Contrarily to SMQS, the weighted majority quorum system (WMQS) was proposed to cope with heterogeneity.
	In WMQS, each process is assigned a weight that is in accordance with the process's
		latency or throughput determined by a monitoring system \cite{AWARE, reliabilityVoting};
		every quorum consists
		of a set of processes such that the sum of their weights is greater than half of the total weight of processes in the system.
	The following example helps to grasp the difference between SMQS and WMQS in systems with heterogeneous latencies and throughput.

	\begin{example}
		Let $p_{1},p_{2},p_{3}$, and $p_{4}$ be the processes comprising the system and $c$ be a client.
		Consider the two following scenarios.
		For the first scenario, assume that the average round-trip latencies between the client and processes $p_{1},p_{2},p_{3}$, and $p_{4}$ are
			$20 ms$, $45 ms$, $100 ms$, and $140 ms$, respectively.
		In another scenario, assume that the throughput of processes $p_{1},p_{2},p_{3}$, and $p_{4}$ are
			$1000$, $800$, $400$, and $200$ operation/sec, respectively.
		Let $1.4$, $1.1$, $0.9$, $0.6$ be the assigned weights by the monitoring system to processes $p_{1},p_{2},p_{3}$, and $p_{4}$, respectively.
		The quorum latency using SMQS is $100 ms$ while using WMQS is $45 ms$ (Figure \ref{fig:smqs-vs-wmqs}).
		The throughput of the system based SMQS and WMQS is $600$ and $800$ operation/sec, respectively\footnote{Throughput is computed using \textit{quoracle} library \cite{readWriteQS}; the code written using \textit{quoracle} to compute throughput can be found in Appendix \ref{appendix-quoracle}.}).
		Both scenarios show the advantage of using the WMQS over SMQS.
		\begin{figure}[hbt!]
			 \centering
			 \begin{subfigure}[t]{0.18\textwidth}
				 \centering
				 \tikzset{every picture/.style={line width=0.75pt}} 

\begin{tikzpicture}[x=0.75pt,y=0.75pt,yscale=-1,xscale=1]

\draw [color={rgb, 255:red, 74; green, 74; blue, 74 }  ,draw opacity=1 ]   (34.08,20.34) -- (38.3,42.13) ;
\draw [shift={(38.68,44.1)}, rotate = 259.04] [fill={rgb, 255:red, 74; green, 74; blue, 74 }  ,fill opacity=1 ][line width=0.08]  [draw opacity=0] (4.8,-1.2) -- (0,0) -- (4.8,1.2) -- cycle    ;
\draw [color={rgb, 255:red, 208; green, 2; blue, 27 }  ,draw opacity=1 ]   (34.28,16.4) -- (83.56,16.4) ;
\draw [color={rgb, 255:red, 208; green, 2; blue, 27 }  ,draw opacity=1 ]   (34.28,18.04) -- (34.28,14.76) ;
\draw [color={rgb, 255:red, 208; green, 2; blue, 27 }  ,draw opacity=1 ]   (83.56,17.85) -- (83.56,14.58) ;

\draw [color={rgb, 255:red, 74; green, 74; blue, 74 }  ,draw opacity=1 ]   (34.08,20.34) -- (45.79,58.37) ;
\draw [shift={(46.38,60.28)}, rotate = 252.9] [fill={rgb, 255:red, 74; green, 74; blue, 74 }  ,fill opacity=1 ][line width=0.08]  [draw opacity=0] (4.8,-1.2) -- (0,0) -- (4.8,1.2) -- cycle    ;
\draw [color={rgb, 255:red, 74; green, 74; blue, 74 }  ,draw opacity=1 ]   (38.68,44.1) -- (43.58,22.22) ;
\draw [shift={(44.02,20.26)}, rotate = 462.61] [fill={rgb, 255:red, 74; green, 74; blue, 74 }  ,fill opacity=1 ][line width=0.08]  [draw opacity=0] (4.8,-1.2) -- (0,0) -- (4.8,1.2) -- cycle    ;
\draw [color={rgb, 255:red, 74; green, 74; blue, 74 }  ,draw opacity=1 ]   (34.08,20.34) -- (70.42,87.65) ;
\draw [shift={(71.38,89.41)}, rotate = 241.63] [fill={rgb, 255:red, 74; green, 74; blue, 74 }  ,fill opacity=1 ][line width=0.08]  [draw opacity=0] (4.8,-1.2) -- (0,0) -- (4.8,1.2) -- cycle    ;
\draw [color={rgb, 255:red, 74; green, 74; blue, 74 }  ,draw opacity=1 ]   (59.02,74.93) -- (82.95,22.29) ;
\draw [shift={(83.78,20.47)}, rotate = 474.45] [fill={rgb, 255:red, 74; green, 74; blue, 74 }  ,fill opacity=1 ][line width=0.08]  [draw opacity=0] (4.8,-1.2) -- (0,0) -- (4.8,1.2) -- cycle    ;
\draw [color={rgb, 255:red, 74; green, 74; blue, 74 }  ,draw opacity=1 ]   (34.08,20.34) -- (58.19,73.11) ;
\draw [shift={(59.02,74.93)}, rotate = 245.45] [fill={rgb, 255:red, 74; green, 74; blue, 74 }  ,fill opacity=1 ][line width=0.08]  [draw opacity=0] (4.8,-1.2) -- (0,0) -- (4.8,1.2) -- cycle    ;
\draw [color={rgb, 255:red, 74; green, 74; blue, 74 }  ,draw opacity=1 ]   (46.38,60.28) -- (56.03,22.2) ;
\draw [shift={(56.52,20.26)}, rotate = 464.22] [fill={rgb, 255:red, 74; green, 74; blue, 74 }  ,fill opacity=1 ][line width=0.08]  [draw opacity=0] (4.8,-1.2) -- (0,0) -- (4.8,1.2) -- cycle    ;
\draw    (24.8,75.24) -- (108.8,75.24) ;
\draw    (24.8,44.74) -- (108.8,44.74) ;
\draw    (24.8,60.47) -- (108.8,60.47) ;
\draw    (24.8,20.21) -- (108.8,20.21) ;
\draw    (24.8,89.93) -- (108.8,89.93) ;
\draw [color={rgb, 255:red, 74; green, 74; blue, 74 }  ,draw opacity=1 ]   (71.38,89.41) -- (101.53,22.21) ;
\draw [shift={(102.35,20.39)}, rotate = 474.17] [fill={rgb, 255:red, 74; green, 74; blue, 74 }  ,fill opacity=1 ][line width=0.08]  [draw opacity=0] (4.8,-1.2) -- (0,0) -- (4.8,1.2) -- cycle    ;

\draw (15.78,20.63) node [anchor=east] [inner sep=0.75pt]  [font=\fontsize{0.88em}{1.06em}\selectfont]  {$c$};
\draw (6.78,85.72) node [anchor=north west][inner sep=0.75pt]  [font=\fontsize{0.88em}{1.06em}\selectfont]  {$p_{4}$};
\draw (6.78,70.33) node [anchor=north west][inner sep=0.75pt]  [font=\fontsize{0.88em}{1.06em}\selectfont]  {$p_{3}$};
\draw (6.78,54.94) node [anchor=north west][inner sep=0.75pt]  [font=\fontsize{0.88em}{1.06em}\selectfont]  {$p_{2}$};
\draw (6.78,40.55) node [anchor=north west][inner sep=0.75pt]  [font=\fontsize{0.88em}{1.06em}\selectfont]  {$p_{1}$};
\draw (40.95,4.04) node [anchor=north west][inner sep=0.75pt]  [font=\fontsize{0.88em}{1.06em}\selectfont]  {$100\ ms$};

\end{tikzpicture}
				 \caption{SMQS}
				 \label{fig:smqs}
			 \end{subfigure}
			 \hfill
			 \begin{subfigure}[t]{0.27\textwidth}
				 \centering
				 \tikzset{every picture/.style={line width=0.75pt}} 

\begin{tikzpicture}[x=0.75pt,y=0.75pt,yscale=-1,xscale=1]

\draw [color={rgb, 255:red, 74; green, 74; blue, 74 }  ,draw opacity=1 ]   (96.08,20.34) -- (100.3,42.13) ;
\draw [shift={(100.68,44.1)}, rotate = 259.04] [fill={rgb, 255:red, 74; green, 74; blue, 74 }  ,fill opacity=1 ][line width=0.08]  [draw opacity=0] (4.8,-1.2) -- (0,0) -- (4.8,1.2) -- cycle    ;
\draw [color={rgb, 255:red, 208; green, 2; blue, 27 }  ,draw opacity=1 ]   (96.28,16.4) -- (118.8,16.4) ;
\draw [color={rgb, 255:red, 208; green, 2; blue, 27 }  ,draw opacity=1 ]   (96.28,18.04) -- (96.28,14.76) ;
\draw [color={rgb, 255:red, 208; green, 2; blue, 27 }  ,draw opacity=1 ]   (118.8,17.85) -- (118.8,14.58) ;

\draw [color={rgb, 255:red, 74; green, 74; blue, 74 }  ,draw opacity=1 ]   (96.08,20.34) -- (107.79,58.37) ;
\draw [shift={(108.38,60.28)}, rotate = 252.9] [fill={rgb, 255:red, 74; green, 74; blue, 74 }  ,fill opacity=1 ][line width=0.08]  [draw opacity=0] (4.8,-1.2) -- (0,0) -- (4.8,1.2) -- cycle    ;
\draw [color={rgb, 255:red, 74; green, 74; blue, 74 }  ,draw opacity=1 ]   (100.68,44.1) -- (105.58,22.22) ;
\draw [shift={(106.02,20.26)}, rotate = 462.61] [fill={rgb, 255:red, 74; green, 74; blue, 74 }  ,fill opacity=1 ][line width=0.08]  [draw opacity=0] (4.8,-1.2) -- (0,0) -- (4.8,1.2) -- cycle    ;
\draw [color={rgb, 255:red, 74; green, 74; blue, 74 }  ,draw opacity=1 ]   (96.08,20.34) -- (132.42,87.65) ;
\draw [shift={(133.38,89.41)}, rotate = 241.63] [fill={rgb, 255:red, 74; green, 74; blue, 74 }  ,fill opacity=1 ][line width=0.08]  [draw opacity=0] (4.8,-1.2) -- (0,0) -- (4.8,1.2) -- cycle    ;
\draw [color={rgb, 255:red, 74; green, 74; blue, 74 }  ,draw opacity=1 ]   (121.02,74.93) -- (144.95,22.29) ;
\draw [shift={(145.78,20.47)}, rotate = 474.45] [fill={rgb, 255:red, 74; green, 74; blue, 74 }  ,fill opacity=1 ][line width=0.08]  [draw opacity=0] (4.8,-1.2) -- (0,0) -- (4.8,1.2) -- cycle    ;
\draw [color={rgb, 255:red, 74; green, 74; blue, 74 }  ,draw opacity=1 ]   (96.08,20.34) -- (120.19,73.11) ;
\draw [shift={(121.02,74.93)}, rotate = 245.45] [fill={rgb, 255:red, 74; green, 74; blue, 74 }  ,fill opacity=1 ][line width=0.08]  [draw opacity=0] (4.8,-1.2) -- (0,0) -- (4.8,1.2) -- cycle    ;
\draw [color={rgb, 255:red, 74; green, 74; blue, 74 }  ,draw opacity=1 ]   (108.38,60.28) -- (118.03,22.2) ;
\draw [shift={(118.52,20.26)}, rotate = 464.22] [fill={rgb, 255:red, 74; green, 74; blue, 74 }  ,fill opacity=1 ][line width=0.08]  [draw opacity=0] (4.8,-1.2) -- (0,0) -- (4.8,1.2) -- cycle    ;
\draw    (19.73,75.24) -- (167.1,75.24) ;
\draw    (19.73,44.74) -- (167.1,44.74) ;
\draw    (19.73,60.47) -- (167.1,60.47) ;
\draw    (19.73,20.21) -- (167.1,20.21) ;
\draw    (19.73,89.93) -- (167.1,89.93) ;
\draw [color={rgb, 255:red, 74; green, 74; blue, 74 }  ,draw opacity=1 ]   (133.38,89.41) -- (163.53,22.21) ;
\draw [shift={(164.35,20.39)}, rotate = 474.17] [fill={rgb, 255:red, 74; green, 74; blue, 74 }  ,fill opacity=1 ][line width=0.08]  [draw opacity=0] (4.8,-1.2) -- (0,0) -- (4.8,1.2) -- cycle    ;

\draw (12.78,20.63) node [anchor=east] [inner sep=0.75pt]  [font=\fontsize{0.88em}{1.06em}\selectfont]  {$c$};
\draw (3.78,85.72) node [anchor=north west][inner sep=0.75pt]  [font=\fontsize{0.88em}{1.06em}\selectfont]  {$p_{4}$};
\draw (3.78,70.33) node [anchor=north west][inner sep=0.75pt]  [font=\fontsize{0.88em}{1.06em}\selectfont]  {$p_{3}$};
\draw (3.78,54.94) node [anchor=north west][inner sep=0.75pt]  [font=\fontsize{0.88em}{1.06em}\selectfont]  {$p_{2}$};
\draw (3.78,40.55) node [anchor=north west][inner sep=0.75pt]  [font=\fontsize{0.88em}{1.06em}\selectfont]  {$p_{1}$};
\draw (91.95,4.04) node [anchor=north west][inner sep=0.75pt]  [font=\fontsize{0.88em}{1.06em}\selectfont]  {$45\ ms$};
\draw (25.75,32.94) node [anchor=north west][inner sep=0.75pt]  [font=\footnotesize] [align=left] {$\displaystyle weight=1.4$};
\draw (25.75,47.86) node [anchor=north west][inner sep=0.75pt]  [font=\footnotesize] [align=left] {$\displaystyle weight=1.1$};
\draw (25.75,62.78) node [anchor=north west][inner sep=0.75pt]  [font=\footnotesize] [align=left] {$\displaystyle weight=0.9$};
\draw (25.75,77.69) node [anchor=north west][inner sep=0.75pt]  [font=\footnotesize] [align=left] {$\displaystyle weight=0.6$};

\end{tikzpicture}
				 \caption{WMQS}
				 \label{fig:wmqs}
			 \end{subfigure}
			\caption{The quorum latency of SMQS vs. WMQS}
			\label{fig:smqs-vs-wmqs}
			\vspace{-0.3cm}
		\end{figure}
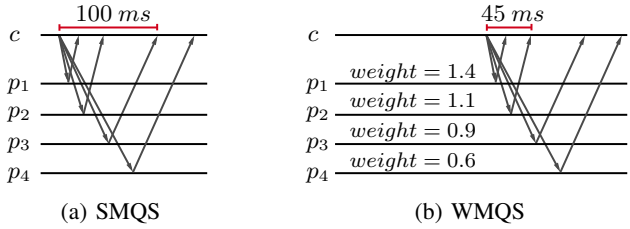
	\end{example}
	Although the WMQS
		improves the quorum latency in contrast to the SMQS,
		it has a significant drawback for real, dynamic, and long-lived systems,
		where the latencies and throughput of processes might change over time.
	Indeed, using time-invariant weights is not suitable for such systems, so the processes' weights must be reassigned over time.
	However, reassigning processes' weights is a challenging problem in dynamic asynchronous failure-prone systems due to the following requirements.

	\begin{itemize}[align=parleft,leftmargin=*]
		\item[a)] \textit{Guaranteeing the atomicity property.} Weight reassignments and r/w operations might be concurrent.
			If it is the case, some r/w operations might be performed based on the \textit{most up-to-date} weights,
				and others might be performed based on previous weights, not ensuring thus the atomicity property.
		\item[b)] \textit{Guaranteeing the liveness of atomic storage.} If one allows arbitrary weight reassignment, the system's liveness might not be guaranteed.
	For instance, in Figure \ref{fig:wmqs}, assume that process $p_{1}$'s weight is reassigned to $2.7$
		while the weights of other processes are not reassigned.
	If process $p_{1}$ fails, no quorum can be constituted, leading to the loss of the system's liveness.
		Note that the other processes' weights cannot be reassigned anymore due to the asynchrony of the system.
		Indeed,	process $p_{1}$ can be slow, and by reassigning other processes' weights, two disjoint weighted quorums might be constituted.
		\item[c)] \textit{Demanding a consensus-free and wait-free solution.} For reassigning weights, a consensus-based protocol or similar primitives cannot be used
			because it is known that consensus is not solvable in asynchronous failure-prone systems \cite{fLP}.
		Besides, the ABD protocol \cite{abd} showed that atomic storage could be implemented in static asynchronous systems
			in a wait-free \cite{wait-free} manner and without requiring consensus.
		\item[d)] \textit{Efficiency.} To have an efficient storage system,
			it is required to separate weight reassignment protocol from r/w protocols \cite{jehl2017case}.
	\end{itemize}

	 Atomic storage with consensus-based weight reassignment protocols (e.g., \cite{dynVotingAlgMaintainingConsistency, mutualExclusionDynamicVoteRessignment})
		does not satisfy the third requirement.
	Likewise, if consensus-based reconfiguration protocols (like RAMBO \cite{rambo}) are adapted to be used as a
		weight reassignment protocol, the third requirement is not satisfied.
	 Consensus-free reconfiguration protocols, like DynaStore \cite{dynAtomicStorageWithoutCons}, SpSn \cite{elasticReconf},
		and FreeStore \cite{effModConsensus-freeFST} (more protocols can be found in \cite{reconfigureASTutorial, spiegelman2017dynamic, jehl2017case}),
		proposed to change the set of
		processes that compromise the system
		by using two special functions: \textit{join} and \textit{leave}.
    Servers can join/leave the system by calling these functions.
   	Such protocols can be adapted to be used as a consensus-free solution for reassigning the processes' weights.
    To do so, one can change \textit{join} and \textit{leave} functions to \textit{increase} and \textit{decrease} functions, respectively
	such that each server can request to increase/decrease its weight using \textit{increase}/\textit{decrease} functions.
	However, such protocols might create unacceptable states in which the liveness of atomic storage cannot be guaranteed (see Example \ref{example:2}).

	\begin{example}\label{example:2}
		Let $p_{1},p_{2},p_{3}$, and $p_{4}$ be the processes comprising the system;
			also, let the initial weight of each process be one (Figure \ref{fig:unacceptable-states}).
		Two concurrent requests $increase(p_{1}, 1.4)$ and $decrease(p_{4},0.7)$ are issued by processes $p_{1}$ and $p_{4}$, respectively,
			to increase $p_{1}$'s weight by 1.4 and to decrease $p_{4}$'s weight by 0.7.
		Each request creates an intermediate (auxiliary) state.
		Although each of created intermediate states is acceptable, their combination is unacceptable
			because the system might not be live (consider the case when process $p_{1}$ fails).
		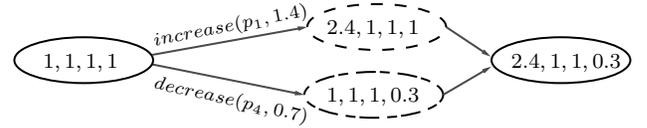
\begin{figure}[hbt!]
			 \centering
			 \tikzset{every picture/.style={line width=0.75pt}} 

\begin{tikzpicture}[x=0.75pt,y=0.75pt,yscale=-1,xscale=1]

\draw [color={rgb, 255:red, 74; green, 74; blue, 74 }  ,draw opacity=1 ]   (70.18,52.72) -- (145.22,67.76) ;
\draw [shift={(147.18,68.15)}, rotate = 191.33] [fill={rgb, 255:red, 74; green, 74; blue, 74 }  ,fill opacity=1 ][line width=0.08]  [draw opacity=0] (4.8,-1.2) -- (0,0) -- (4.8,1.2) -- cycle    ;
\draw   (0.93,50.65) .. controls (0.93,44.26) and (16.6,39.08) .. (35.93,39.08) .. controls (55.26,39.08) and (70.93,44.26) .. (70.93,50.65) .. controls (70.93,57.04) and (55.26,62.22) .. (35.93,62.22) .. controls (16.6,62.22) and (0.93,57.04) .. (0.93,50.65) -- cycle ;
\draw  [dash pattern={on 4.5pt off 4.5pt}] (148.68,33.9) .. controls (148.68,27.51) and (164.35,22.33) .. (183.68,22.33) .. controls (203.01,22.33) and (218.68,27.51) .. (218.68,33.9) .. controls (218.68,40.29) and (203.01,45.47) .. (183.68,45.47) .. controls (164.35,45.47) and (148.68,40.29) .. (148.68,33.9) -- cycle ;
\draw  [dash pattern={on 3.75pt off 3pt on 7.5pt off 1.5pt}] (147.18,68.15) .. controls (147.18,61.76) and (162.85,56.58) .. (182.18,56.58) .. controls (201.51,56.58) and (217.18,61.76) .. (217.18,68.15) .. controls (217.18,74.54) and (201.51,79.72) .. (182.18,79.72) .. controls (162.85,79.72) and (147.18,74.54) .. (147.18,68.15) -- cycle ;
\draw [color={rgb, 255:red, 74; green, 74; blue, 74 }  ,draw opacity=1 ]   (70.43,47.97) -- (146.71,34.26) ;
\draw [shift={(148.68,33.9)}, rotate = 529.81] [fill={rgb, 255:red, 74; green, 74; blue, 74 }  ,fill opacity=1 ][line width=0.08]  [draw opacity=0] (4.8,-1.2) -- (0,0) -- (4.8,1.2) -- cycle    ;
\draw   (241.6,50.65) .. controls (241.6,44.26) and (257.27,39.08) .. (276.6,39.08) .. controls (295.93,39.08) and (311.6,44.26) .. (311.6,50.65) .. controls (311.6,57.04) and (295.93,62.22) .. (276.6,62.22) .. controls (257.27,62.22) and (241.6,57.04) .. (241.6,50.65) -- cycle ;
\draw [color={rgb, 255:red, 74; green, 74; blue, 74 }  ,draw opacity=1 ]   (218.68,33.9) -- (239.53,47.77) ;
\draw [shift={(241.2,48.88)}, rotate = 213.63] [fill={rgb, 255:red, 74; green, 74; blue, 74 }  ,fill opacity=1 ][line width=0.08]  [draw opacity=0] (4.8,-1.2) -- (0,0) -- (4.8,1.2) -- cycle    ;
\draw [color={rgb, 255:red, 74; green, 74; blue, 74 }  ,draw opacity=1 ]   (217.18,68.15) -- (239.81,54.89) ;
\draw [shift={(241.53,53.88)}, rotate = 509.62] [fill={rgb, 255:red, 74; green, 74; blue, 74 }  ,fill opacity=1 ][line width=0.08]  [draw opacity=0] (4.8,-1.2) -- (0,0) -- (4.8,1.2) -- cycle    ;

\draw (14,46) node [anchor=north west][inner sep=0.75pt]  [font=\footnotesize] [align=left] {$\displaystyle 1,1,1,1$};
\draw (157,29) node [anchor=north west][inner sep=0.75pt]  [font=\footnotesize] [align=left] {$\displaystyle 2.4,1,1,1$};
\draw (157,63) node [anchor=north west][inner sep=0.75pt]  [font=\footnotesize] [align=left] {$\displaystyle 1,1,1,0.3$};
\draw (68.62,35.2) node [anchor=north west][inner sep=0.75pt]  [font=\scriptsize,rotate=-349.06] [align=left] {$\displaystyle increase( p_{1} ,1.4)$};
\draw (71.48,56.08) node [anchor=north west][inner sep=0.75pt]  [font=\scriptsize,rotate=-12.6] [align=left] {$\displaystyle decrease( p_{4} ,0.7)$};
\draw (250,46) node [anchor=north west][inner sep=0.75pt]  [font=\footnotesize] [align=left] {$\displaystyle 2.4,1,1,0.3$};

\end{tikzpicture}
			\caption{An example to show that reconfiguration protocols (like DynaStore, SpSn, and FreeStore) might
				create unacceptable states if they are adapted to be used as a weight reassignment protocol.
				Each oval (resp. dashed oval) is a state (resp. an intermediate state); $i$th number of each state determines $p_{i}$'s weight, where $i \in \{1,2,3,4\}$.}
			\label{fig:unacceptable-states}
	\vspace{-0.3cm}
		\end{figure}
	\end{example}

	SmartMerge \cite{smartmerge} is the only consensus-free reconfiguration protocol that avoids creating unacceptable states;
		however, for each r/w operation, it requires communicating with a quorum of
		processes to find the most up-to-date configuration.
	Therefore, it might incur significant performance losses in terms of latency and throughput, i.e.,
		it does not satisfy the fourth requirement since the r/w operations are not completely separated from the reconfiguration protocol.
	If SmartMerge is adapted to be used as a weight reassignment protocol, we still have the same problem.

	To the best of our knowledge, no protocol is presented explicitly to solve weight reassignment
		in a consensus-free manner for atomic storage.
	This paper presents a novel and efficient consensus-free weight reassignment protocol that
		autonomously reassigns the processes' weights for atomic storage (the atomic storage is based on the ABD protocol).
	The weight reassignment protocol avoids creating unacceptable states.
	In contrast to other solutions, the distinguishing feature of our protocol is that for executing each r/w operation,
		it is not required to communicate with a quorum of processes to find the most up-to-date processes' weights leading to efficiency improvement.
	To evaluate the performance of an atomic storage based
		on our weight reassignment protocol,
		we compared our approach to atomic storages based on
		(1) the (static) ABD that uses an SMQS,
		(2) RAMBO, and
		(3) SmartMerge.
	Our experimental results show that our approach is
		38$\%$, 17$\%$, 27$\%$ more efficient than
		the (static) ABD, RAMBO, and SmartMerge, respectively.

	\vspace{0.5em}
	\noindent\textbf{Organization of the paper.}
		Section \ref{sec:preliminaries} presents the system model and some preliminary definitions and properties used in the paper.
		In Section \ref{sec:dynamic-weight-reassignment}, we describe our weight reassignment protocol.
		We present the dynamic atomic storage that utilizes our weight reassignment protocol in Section \ref{sec:storage-system}.
		The performance evaluation is shown in Section \ref{sec:evaluation}.
		We present the conclusion and future work in Section \ref{sec:conclusion}.

\section{Preliminaries}\label{sec:preliminaries}
	In this section, we present the system model of our paper.
	Also, we present the preliminary definitions and properties of our weight reassignment protocol and dynamic
		atomic storage system.

	\subsection*{System Model}\label{subsec:sysmte_model}
		We consider a distributed system composed by two non-overlapping sets of processes--
			a finite set of $n$ servers $S = \{s_{1},s_{2},\dots, s_{n}\}$
			and an infinite set of clients $C = \{c_{1}, c_{2}, \dots\}$.
		Each process has a unique identifier.
		Every client or server knows the set of servers.
		Clients access the storage system provided by servers by executing r/w operations.
		The processes communicate by message passing, and the links reliably connect all pairs of processes.
		Processes are prone to crash failures.
		A process is called \textit{correct} if it is not crashed.
		The system is asynchronous, i.e., we make no assumptions about processing times or
			message transmission delays.
		However, each process has access to a local clock;
			processes' local clocks are not synchronized and do not have any bounds on their drifts, being nothing more than counters that keep increasing.
		The interactions between processes are assumed to take place over a timespan $\mathcal{T} \subset \mathbb{R}^{+}$ called the lifetime of the system.

	\subsection*{Weight Reassignment Definitions and Properties}
		\noindent\textbf{Views.}
		During the system's lifetime, a sequence of views $\sigma = \langle v_{0}, v_{1}, \dots \rangle$ is installed in the system
			to reassign servers' weights.
		The system starts in view $v_{0}$ called the \textit{initial view}.
		The successor (resp. predecessor) of any view $v_{k}$ for $0 \leq k$ (resp. $1 \leq k$)
			is $v_{k+1} = v_{k}.succ$ (resp. $v_{k-1} = v_{k}.pred$).
		We say a view $v$ is installed in the system if a few correct servers consider $v$ as their current view (see definition below).
		We denote the current view of any server $s_{i}$ by $s_{i}.cview$.
		Note that the current views of servers might be different from the installed view in the system.
		When a non-initial view $v_{k+1}$ ($0 \leq k$) is installed, we say that $v_{k+1}.pred$ was uninstalled from the system.
		At any time $t \in \mathcal{T}$, we define $lastview$ to be the last view installed in the system.
		Since $lastview$ is the last installed view in the system, $lastview.succ = \perp$.
		The weights of servers are not reassigned during any view $v_{k}$ ($0 \leq k$)
			and might be reassigned at the time of uninstalling $v_{k}$ and installing $v_{k+1}$.

		\vspace{0.5em}
		\noindent\textbf{Installing a view.} To install a view in the system, at least one server should request it.
		Each server $s_{i}$ can only request to install view $s_{i}.cview.succ$.
		To do so, $s_{i}$ sends a message $\langle \textbf{change\_view}, s_{i}.cview.succ \rangle$ to other servers.
		Each server $s_{j}$ sends each received message $\langle \textbf{change\_view}, v\rangle$ to other servers
			if it had not sent such a message previously.
		Then, each server $s_{i}$ sends message $\langle \textbf{state\_update}, *,*,*,s_{i}.cview, w\rangle$\footnote{We use * for a parameter when its value is not important.} to other servers,
			where the last parameter stands for $s_{i}$'s weight in $s_{i}.cview$ (we explain the algorithm for changing the views in further detail in Section \ref{sec:dynamic-weight-reassignment}).
		Each server $s_{i}$ can install view $s_{i}.cview.succ$
			as soon as receiving messages $\langle \textbf{state\_update}, *,*,*,v, w \rangle$
			from a weighted majority of servers with views $v = s_{i}.cview$.
		As soon as at least one server installs a view $v$ such that $\forall \ s \in S : s.cview \leq v$,
			we say that view $v$ is installed in the system.

		\vspace{0.5em}
		\noindent\textbf{Comparing two views.}
		We say that view $w$ is \textit{more up-to-date} than view $v$ if the
			following recursive function returns $yes$ by passing $(v, w)$ as input.
		We use the notation $v < w$ to state that view $w$ is \textit{more up-to-date} than view $v$.

		\begin{center}
			\fbox{\parbox{0.7\linewidth}{
				\begin{algorithmic}
					\STATEx{\hspace{-1em}\textbf{function} $more\_up\_to\_date(v, w)$}
						\STATE{$v \leftarrow v.succ$}
						\STATE{\textbf{if} $v = w$ \textbf{then} return $yes$}
						\STATE{\textbf{else if} $v = \perp$ \textbf{then} return $no$}
						\STATE{\textbf{else} return $more\_up\_to\_date(v, w)$}
				\end{algorithmic}
			}}
		\end{center}

		\vspace{0.5em}
		\noindent\textbf{Our protocol's assumptions and properties.} The following
			assumptions and properties are required in our weight reassignment protocol.
		From now on, $\mathbb{wl}$, $\mathbb{wu}$, $\mathbb{W_{T}}$, and $f$ state the lower bound of servers' weights,
			the upper bound of servers' weights, the total weight of servers, and the maximum number of failed servers, respectively.
		\begin{assumption}\label{property:initial-weights}
			The initial weight of each server is equal to one.
			Formally, $\forall \ s_{i} \in S : s_{i}.v_{0}.weight = 1$.
		\end{assumption}
		\begin{assumption}\label{assumption:wl-wu}
			The values of $\mathbb{wl}$ and $\mathbb{wu}$ are $n/(2\times(n-f))$ and $n/(2f)$, respectively.
		\end{assumption}
		\begin{assumption}\label{assumption:quorum-size}
			In each quorum, the total weight of servers is greater than $n/2$.
		\end{assumption}
		\begin{assumption}\label{assumption:finite-view-change}
			The number of views that are requested to be installed in the system is finite.
			Formally, $|\sigma| = m$, where $m \in \mathbb{N}$ and $\sigma$ is the sequence of installed views.
		\end{assumption}
		Assumption \ref{assumption:finite-view-change} is used in all reconfiguration protocols presented for asynchronous systems such as RAMBO, DynaStore, SpSn, FreeStore, and SmartMerge.
		Its reason is that it is impossible to reconfigure a storage system infinitely many times while guaranteeing the liveness of the storage system \cite{spiegelman2017liveness}.

		\begin{property}\label{property:total-weight}
			The total weight of servers is bounded by $n$ in any view $v \in \sigma$.
			Formally, $\forall \ v \in \sigma, \sum_{s_{i} \in S}s_{i}.v.weight \leq n$.
		\end{property}
		\begin{property}\label{property:weight-lower-upper-bound}
			The weight of each server in every view should be greater than $\mathbb{wl}$ and less than $\mathbb{wu}$.
			Formally, $\forall \ v \in \sigma, \forall \ s_{i} \in S : \mathbb{wl} < s_{i}.v.weight < \mathbb{wu}$.
		\end{property}
		If a system relies on Assumption \ref{assumption:wl-wu}, Assumption \ref{assumption:quorum-size}, Property \ref{property:total-weight}, and Property \ref{property:weight-lower-upper-bound},
			we can show that there is a quorum of correct servers, even if $f$ servers crash (in the worst case, $f$ servers crash so there are $n-f$ correct servers, each one with weight $\mathbb{wl}$;
			since $(n-f) \times \mathbb{wl} > n/2$, at least a quorum of servers can be constituted).
		Consequently, unacceptable states (like in Example \ref{example:2}) are not created.


		\vspace{0.5em}
		\noindent\textbf{Assumptions related to WMQS.} The following
			assumptions are required to have performance gains by using WMQS.
		\begin{assumption}\label{property:fn}
			During the system's lifetime: $2f + 1 \leq n$.
		\end{assumption}
		This assumption is the same as the one used in other weight reassignment protocols, like WHEAT \cite{geoReplicatedSMR} and AWARE \cite{AWARE}.
		Indeed, this assumption states that there are a few additional spare servers, enabling the system to make progress without needing to access a majority of servers.
%
		\begin{assumption}\label{property:qcnfwu}
			Constant $\mathbb{wu}$ should be defined in such a way that
				$1 \leq \mathbb{wu}$.
		\end{assumption}
		The goal of WMQS is to constitute quorums with a minority of high-weighted servers to improve performance.
		To constitute a quorum with a minority of servers, it is required that $1 \leq \mathbb{wu}$.
		The reason for having such a requirement is as follows.
		Due to Assumption \ref{assumption:quorum-size}, the total weight of servers is greater than $n/2$ in each quorum.
		To have a quorum with a minority of servers, it is necessary (but not sufficient) to have at least one server with a weight greater than equal to $1$.

		\vspace{0.5em}
		\noindent\textbf{Monitoring system.}
		In order to reassign servers' weights, the latencies of the server to server and client to server communications should be monitored.
		To this end, each server uses a local monitor module (like the one presented in AWARE \cite{AWARE}) that is responsible for evaluating and gathering information about latencies and giving scores to servers.
		We denote the latency score of server $s_{k}$ computed using the monitoring system of server $s_{i}$ by $s_{i}.lscores.s_{k}$ ($1 \leq k \leq n$).
		Note that it might be possible that $s_{i}.lscores.s_{k} \neq s_{j}.lscores.s_{k}$, at any time $t$.
		Any server $s_{i}$ can compare its latency with another server $s_{j}$ using latency scores.
		For instance, from server $s_{i}$'s point of view,
			the latency of server $s_{j}$ is greater than $s_{i}$'s latency if $s_{i}.lscores.s_{i} <  s_{i}.lscores.s_{j}$.

	\subsection*{Dynamic Storage Definitions and Properties}\label{subsec:dyn_wei_storage_properties}
		\noindent\textbf{Views vs. r/w operations.} At any time $t \in \mathcal{T}$, r/w operations
			can only be executed in view $lastview$.
		At the time of uninstalling any view $v$, r/w operations are disabled on servers with view $v$.
		The operations are enabled after installing view $v.succ$.

		\begin{definition}[Atomic register \cite{interprocessComII}]\label{def:atomic-register}
			Assume two read operations $r_{1}$ and $r_{2}$ executed by correct clients.
			Consider that $r_{1}$ terminates before $r_{2}$ initiates.
			If $r_{1}$ reads a value $\alpha$ from register $R$,
				then either $r_{2}$ reads $\alpha$ or $r_{2}$ reads a more up-to-date value than $\alpha$.
		\end{definition}

		\vspace{0.5em}
		\noindent\textbf{Dynamic storage.} A dynamic storage satisfies the following properties:
			(1) the r/w protocols should implement an atomic register (Definition \ref{def:atomic-register}),
			(2) every r/w operation executed by a correct client eventually terminates,
			(3) the r/w operations that are disabled on servers to install a view will eventually be enabled,
			(4) if any server $s_{i}$ installs a non-initial view $v$, some server has requested to install view $v$,
			(5) reassigning weights are possible during the lifetime of the system, i.e.,
				the weight reassignment protocol satisfies the liveness property.



\section{Weight Reassignment Protocol}\label{sec:dynamic-weight-reassignment}
    In this section, we describe our weight reassignment protocol.
    The protocol has two essential dependent algorithms: \textit{pairwise weight reassignment} and \textit{view changer}.
    In the following, we describe these algorithms starting with \textit{pairwise weight reassignment}.
    Then, we present the main properties of the weight reassignment protocol.

    \subsection*{Pairwise Weight Reassignment}\label{subsec:pairwise-weight-assignment}
        The default weight of servers in each view is one.
        However, for each succeeding view, any server might reassign its default weight. 
        As a result, each server might have different weights in distinct views.
        Each server $s_{i}$ to reassign its default weight associated with
            a view $v$ ($s_{i}.cview < v$) should participate in at least one pairwise weight reassignment.
        Pairwise weight reassignment is an algorithm in which two servers collaborate to reassign their weights associated with a view.
        Each pairwise weight reassignment $pwr$ is characterized by a quadruple $(pwr\_receiver, pwr\_sender, w, v)$ such that
            for view $v$, a server called \textit{pwr\_sender} decreases its weight associated with view $v$ by weight $w$ and sends $w$ in a message to
            another server called \textit{pwr\_receiver} that has lower latency;
            \textit{pwr\_receiver} increases its weight associated with $v$ by weight $w$ after receiving the message containing $w$.
        In this way, the servers that have lower latencies might become high-weighted servers leading to improving the performance.

        Each server should satisfy the lower and upper bounds defined in Assumption \ref{assumption:wl-wu} and Property \ref{property:weight-lower-upper-bound}.
        In other words, a server does not participate in a pairwise weight reassignment if its weight does not meet the lower and upper bounds.
        Besides, the total weight of the system is not changed by reassigning the weights in a pairwise manner;
            therefore, Property \ref{property:total-weight} can be satisfied as well.
        Consequently, unacceptable states cannot be created.

        The pseudo-code of pairwise weight reassignment can be found in Algorithm \ref{alg:pwr}.
        Each pairwise weight reassignment is started by sending a request issued by a server that wants to be the \textit{pwr\_receiver} (Lines \ref{lbl:alg:pwr:5}-\ref{lbl:alg:pwr:12}).
        For simplicity, we assume that every requested weight is equal to a constant $\epsilon$,
            i.e., for every pairwise weight reassignment $pwr = (*,*,w,*)$, $w = \epsilon$.
        Moreover, we assume that each server $s_{i}$ can participate in a pairwise weight reassignment $pwr$
            as the \textit{pwr\_receiver} if $v = s_{i}.cview.succ$, where $pwr=(s_{i}, *, *, v)$.
        In other words, each server can only be a \textit{pwr\_receiver} for its succeeding view.
        Server $s_{i}$ should meet the following conditions to be allowed to send
            a request to server $s_{j}$ to start a pairwise weight reassignment for view $v$.

        \begin{itemize}[leftmargin=*]
            \setlength{\itemindent}{1.6em}
            \item[C1R)] View $v$ should be equal to the succeeding view of server $s_{i}$.
                Formally, $s_{i}.cview.succ = v$ (Lines \ref{lbl:alg:pwr:6} and \ref{lbl:alg:pwr:12}).
            \item[C2R)] Server $s_{i}$ has not participated in any operation related to $s_{i}.cview.succ$.
                This can be ensured by a variable of Algorithm \ref{alg:change-view} (Lines \ref{lbl:alg:pwr:7}-\ref{lbl:alg:pwr:8}).
            \item[C3R)] Each server is allowed to send a request to another server that has a greater latency score.
                Formally, $s_{i}.lscores.s_{i} < s_{i}.lscores.s_{j}$ (Line \ref{lbl:alg:pwr:9}).
            \item[C4R)] Each server is allowed to send a request if its weight does not exceed the upper bound $\mathbb{wu}$ defined in Property \ref{property:weight-lower-upper-bound}.
                Formally, $get\_weight(cview.succ) + get\_requested\_weight(cview.succ)$ $+ \epsilon < \mathbb{wu}$.
                Function $get\_weight(v)$ is used to determine the weight associated with view $v$,
                and $get\_requested\_weight(v)$ is a function to determine the total weight of requested weights that their responses have not received yet for view $v$ (Line \ref{lbl:alg:pwr:10}).
        \end{itemize}

        If the above conditions are met, server $s_{i}$
            is allowed to send a request by a message
            $\langle \textbf{propose\_pwr}, s_{i}.cview.succ, \epsilon \rangle$ to sever $s_{j}$ (Line \ref{lbl:alg:pwr:12}).
        By sending this message, we say that server $s_{i}$ proposes a \textit{pairwise weight reassignment} to server $s_{j}$ for view $v = s_{i}.cview.succ$.
        Besides, server $s_{i}$ adds the pairwise weight reassignment to a set $s_{i}.pwr\_requests$ to meet C4R (Line \ref{lbl:alg:pwr:11}).
        Each server $s_{i}$ meets the following conditions by receiving any message $\langle \textbf{propose\_pwr}, v, \epsilon \rangle$ from server $s_{j}$.

        \begin{itemize}[leftmargin=*]
            \setlength{\itemindent}{1.5em}
            \item[C1S)] View $v$ should be as up-to-date as $s_{i}.cview.succ$.
                Formally, $s_{i}.cview.succ \leq v$ (Line \ref{lbl:alg:pwr:15}).
            \item[C2S)] Server $s_{i}$ has not participated in any operation related to $s_{i}.v$ (Lines \ref{lbl:alg:pwr:16}-\ref{lbl:alg:pwr:17}).
            \item[C3S)] The latency score of $s_{i}$ should be greater than $s_{j}$.
                Formally, $s_{i}.lscores.s_{j} < s_{i}.lscores.s_{i}$ (Line \ref{lbl:alg:pwr:18}).
            \item[C4S)] Each server is allowed to accept a request if its weight does not get less than $\mathbb{wl}$ to satisfy Property \ref{property:weight-lower-upper-bound}.
                Formally, $\mathbb{wl}$ $< get\_weight(view.succ)$ $- \epsilon$ (Line \ref{lbl:alg:pwr:19}).
        \end{itemize}

        If the above conditions are met, server $s_{i}$ executes a command $pwrs \leftarrow pwrs \cup \{(s_{j}, s_{i}, v, -\epsilon)\}$
            to store the pairwise weight reassignment associated with view $v$ (Line \ref{lbl:alg:pwr:20}).
        Then, server $s_{i}$ sends message $\langle \textbf{accept\_pwr}, v, \epsilon \rangle$ to server $s_{j}$ (Line \ref{lbl:alg:pwr:21}).
        By sending this message, we say that server $s_{i}$ accepts the pairwise weight reassignment.
        Moreover, the pairwise weight assignment terminates for server $s_{i}$.
        For simplicity, we omitted the part that server $s_{i}$ does not accept the pairwise weight reassignment.

        Server $s_{i}$ meets conditions C1R and C2R by receiving any message $\langle \textbf{accept\_pwr}, v, \epsilon \rangle$ from server $s_{j}$ (Lines \ref{lbl:alg:pwr:23}-\ref{lbl:alg:pwr:25}).
        If the conditions are met, server $s_{i}$ executes a command $pwrs \leftarrow pwrs \cup \{(s_{i}, s_{j}, v, \epsilon)\}$ to store the pairwise weight reassignment associated with view $v$ (Line \ref{lbl:alg:pwr:26}).
        Also, server $s_{i}$ removes the terminated pairwise weight reassignment from set $pwr\_requests$ (Line \ref{lbl:alg:pwr:27}).
        At this point, we say that the pairwise weight reassignment terminates for server $s_{i}$.

        \begin{algorithm*}[hbt!]
            \caption{Pairwise weight reassignment - server $s_{i}$}
            \label{alg:pwr}
            \small
            \begin{algorithmic}[1]
                \STATEx{\hspace{-1.65em}\textbf{variables}}
                    \STATE{$pwrs \leftarrow \emptyset$}           \label{lbl:alg:pwr:1}          \COMMENT{a set for storing pairwise weight reassignments}
                    \STATE{$pwr\_requests \leftarrow \emptyset$}    \label{lbl:alg:pwr:2}  \COMMENT{a set for storing requested pairwise weight reassignments}

                \STATEx{\hspace{-1.65em}\textbf{functions}}
                    \STATE{$get\_weight(view) \equiv 1 + sum\big(\{w \ | \ (*,*, v, w) \in pwrs \text{ and } view = v \}\big)$ \hfill $\triangleright$ a function to compute $s_{i}.view.weight$}   
                    \STATE{$get\_requested\_weight(view) \equiv sum\big(\{w \ | \ (*,*, v, w) \in pwr\_requests \text{ and } view = v \}\big)$}

                \STATEx{\hspace{-1.65em}\textbf{while} $forever$}
                    \STATE{\textbf{atomic}}                        \label{lbl:alg:pwr:5}         \COMMENT{atomic execution of lines 6-12; $s_{i}$ tries to be a \textit{pwr\_receiver}}
                    \Indent
                    \vspace{-0.1cm}
                    \STATE{$cview \leftarrow get\_cview()$ of Algorithm \ref{alg:change-view}}    \label{lbl:alg:pwr:6}
                    \STATE{$dirty\_views \leftarrow get\_dirty\_views()$ of Algorithm \ref{alg:change-view}} \label{lbl:alg:pwr:7}
                    \IF{$cview.succ \notin dirty\_views$}  \label{lbl:alg:pwr:8}                                                                                                                \COMMENT{$s_{i}$ should not participate in any operation related to $cview.succ$ (C2R)}
                        \IF{$\exists \ s_{j} : s_{i}.lscores.s_{i} < s_{i}.lscores.s_{j}$} \label{lbl:alg:pwr:9}                                                                                \COMMENT{$s_{i}$ tries to find a server $s_{j}$ with a greater latency score (C3R)}
                            \IF{$get\_weight(cview.succ) + get\_requested\_weight(cview.succ) + \epsilon < \mathbb{wu}$ }                           \label{lbl:alg:pwr:10}                      \COMMENT{$s_{i}$'s weight should not exceed $\mathbb{wu}$ (C4R)}
                                \STATE{$pwr\_requests \leftarrow pwr\_requests \cup \{(s_{i}, s_{j}, cview.succ, \epsilon)\}$}                      \label{lbl:alg:pwr:11}                      \COMMENT{the sent request is stored in set $pwr\_requests$}
                                \STATE{send message $\langle \textbf{propose\_pwr}, cview.succ, \epsilon \rangle$ to server $s_{j}$}                \label{lbl:alg:pwr:12}
                            \ENDIF
                        \ENDIF
                    \ENDIF
                    \EndIndent

                \STATEx{\hspace{-1.65em}\textbf{upon receipt of} message $\langle \textbf{propose\_pwr}, view, \epsilon \rangle$ from server $s_{j}$} \COMMENT{$s_{i}$ is the \textit{pwr\_sender}}
                    \STATE{\textbf{atomic}}                      \label{lbl:alg:pwr:13}           \COMMENT{atomic execution of lines 14-21}
                    \Indent
                    \vspace{-0.1cm}
                    \STATE{$cview \leftarrow get\_cview()$ of Algorithm \ref{alg:change-view}} \label{lbl:alg:pwr:14}
                    \IF{$cview.succ \leq view$}                 \label{lbl:alg:pwr:15}         \COMMENT{$view$ should be as up-to-date as $cview.succ$ (C1S)}
                        \STATE{$dirty\_views \leftarrow get\_dirty\_views()$ of Algorithm \ref{alg:change-view}} \label{lbl:alg:pwr:16}
                        \IF{$view \notin dirty\_views$}       \label{lbl:alg:pwr:17}        \COMMENT{$s_{i}$ should not participate in any operation related to $view$ (C2S)}
                            \IF{$s_{i}.lscores.s_{j} < s_{i}.lscores.s_{i}$}   \label{lbl:alg:pwr:18} \COMMENT{$s_{j}$ should has a lower latency than $s_{i}$ (C3S)}
                                \IF{$\mathbb{wl} < get\_weight(view.succ) - \epsilon$} \label{lbl:alg:pwr:19}                                                                                                           \COMMENT{$s_{i}$'s weight should not get less than $\mathbb{wl}$ (C4S)}
                                    \STATE{$pwrs \leftarrow pwrs \cup \{(s_{j}, s_{i}, view, -\epsilon)\}$}  \label{lbl:alg:pwr:20} \COMMENT{storing the pairwise weight reassignment}
                                    \STATE{send message $\langle \textbf{accept\_pwr}, view, \epsilon \rangle$ to server $s_{j}$} \label{lbl:alg:pwr:21}
                                \ENDIF
                            \ENDIF
                        \ENDIF
                    \ENDIF
                    \EndIndent

                \STATEx{\hspace{-1.65em}\textbf{upon receipt of} message $\langle \textbf{accept\_pwr}, view, \epsilon \rangle$ from server $s_{j}$}    \COMMENT{$s_{i}$ is the \textit{pwr\_receiver}}
                    \STATE{\textbf{atomic}}                            \label{lbl:alg:pwr:22}     \COMMENT{atomic execution of lines 23-27}
                    \Indent
                    \vspace{-0.1cm}
                    \STATE{$cview \leftarrow get\_cview()$ of Algorithm \ref{alg:change-view}} \label{lbl:alg:pwr:23}
                    \IF{$cview.succ = view$}                       \label{lbl:alg:pwr:24}          \COMMENT{meeting Condition C1R}
                        \IF{$view \notin dirty\_views$} \label{lbl:alg:pwr:25}  \COMMENT{$s_{i}$ has not participated in any operation related to $view$}
                            \STATE{$pwrs \leftarrow pwrs \cup \{(s_{i}, s_{j}, view, \epsilon)\}$}                                            \label{lbl:alg:pwr:26}                 \COMMENT{storing the pairwise weight reassignment}
                            \STATE{$pwr\_requests \leftarrow pwr\_requests \ \backslash \ \{(s_{i}, s_{j}, view, \epsilon)\}$}   \label{lbl:alg:pwr:27}        \COMMENT{removing the terminated pwr form $pwr\_requests$}
                        \ENDIF
                    \ENDIF
                    \EndIndent
            \end{algorithmic}
        \end{algorithm*}

    \subsection*{View Changer}\label{subsec:how-to-change-view}
        Each server can request to change the installed view in the system and change its current view
            by using an algorithm called the view changer algorithm.
        Algorithm \ref{alg:change-view} is the pseudo-code of the view changer algorithm.
        In the following, we describe how this algorithm works.

        \vspace{0.5em}
        \noindent\textbf{How servers can request to change a view.}
        Each server $s_{i}$ for each view $v$ has a timeout (Line \ref{line:alg:vc-1} of Algorithm \ref{alg:change-view}).
        When the timeout of view $s_{i}.cview$ finishes, server $s_{i}$ sends a request to other servers
            to change $s_{i}.cview$ to $s_{i}.cview.succ$ and stores
            $s_{i}.cview.succ$ in a set denoted by $s_{i}.rchange\_views$
        (Lines \ref{alg:line:mod}-\ref{alg:line:send-change-view}).
        Such a request is sent by a message $\langle \textbf{change\_view}, s_{i}.cview.succ \rangle$.
        Note that in practice, such a timeout should be big enough so that the views are changed rarely to satisfy Assumption \ref{assumption:finite-view-change}.

        \vspace{0.5em}
        \noindent\textbf{How a server can change its current view.}
        Each server $s_{i}$, by receiving any message $\langle \textbf{change\_view}, view \rangle$,
             stores $view$ in a set denoted by $s_{i}.rchange\_views$ (line \ref{alg:line:change-views}).
        As soon as $s_{i}.cview.succ \in s_{i}.rchange\_views$, server $s_{i}$ starts to change its view
            (line \ref{alg:line:cview.succ-in-change-views}).
        To do so, server $s_{i}$ must do the following steps:
            (S1) Sending message $\langle \textbf{change\_view}, view \rangle$ to other servers if
                it had not been sent yet.
                Set $schange\_views$ is used to store sent messages tagged with \textbf{change\_view}
                    to be sure that a message is not sent more than once (Line \ref{alg:line:end-if-schang}).
             (S2) Disabling r/w operations (Line \ref{alg:line:disable-rw}).
            (S3) Informing Algorithm \ref{alg:pwr} that some operations related to view $s_{i}.cview.succ$ are processing to safety Conditions C2R and C2S (Line \ref{alg:line:add-dirty-views}).
            (S4) Updating the states (registers) of servers (Lines \ref{line:alg-vc-read-state}-\ref{alg:line:end-update-state}).
                Each server has a register;
                    in this step, the registers of servers with view $view$ are synchronized.
                To do so, server $s_{i}$ reads its register state (say, $\varsigma$) (Line \ref{line:alg-vc-read-state}).
                Then, server $s_{i}$ sends message $\langle \textbf{state\_update}, \varsigma, s_{i}.cview, s_{i}.cview.weight \rangle$
                    to other servers (Line \ref{alg:line:start-update-state}).
                Server $s_{i}$ waits until receiving messages $\langle \textbf{state\_update}, *, s_{i}.cview, * \rangle$
                    from a weighted majority of servers (Line \ref{alg:line:wait-weighted-majority}).
                Finally, server $s_{i}$ computes and stores the new state of its register (Lines \ref{alg:line:maxts}-\ref{alg:line:set-ts-cid-val}).
            (S5) Changing the view (Line \ref{alg:line:update-view}).
            (S6) Enabling r/w operations (Line \ref{alg:line:enable-rw}).

        \begin{algorithm*}[hbt!]
            \caption{View changer - server $s_{i}$}
            \label{alg:change-view}
            \small
            \begin{algorithmic}[1]
            \STATEx{\hspace{-1.65em}\textbf{variables}}
                \STATE{$cview \leftarrow v_{0}$, set a timeout for $cview$}    \label{line:alg:vc-1}                    \COMMENT{a variable to store the current view of $s_{i}$}
                \STATE{$rchange\_views \leftarrow \emptyset$}     \COMMENT{a set to store received $change\_views$}
                \STATE{$schange\_views \leftarrow \emptyset$}     \COMMENT{a set to store sent $change\_views$}
                \STATE{$state\_updates \leftarrow \emptyset$}     \COMMENT{a set to store received $state\_updates$}
                \STATE{$dirty\_views \leftarrow \{v_{0}\}$}       \COMMENT{a set to store the views that $s_{i}$ has participated in; it is used to meet Conditions C2R and C2S in Algorithm \ref{alg:pwr}}

            \STATEx{\hspace{-1.65em}\textbf{functions}}
                \STATE{$sum\_weights(view) \equiv sum \big( \{w \ | \ (*,*,*,*,v,w) \in state\_updates \text{ and } view = v \} \big)$}
                \STATE{$maxts(view) \equiv max \big(\{t \ | \ (*,*,t,*,v,*) \in state\_updates \text{ and } view = v \}\big)$}
                \STATE{$maxcid(view,ts) \equiv max\big(\{c \ | \ (*,*,t,c,v,*) \in state\_updates \text{ and } view = v \text{ and } ts = t \}\big)$}
                \STATE{$val\_maxts\_maxcid(view,ts,cid) \equiv \{val \ | \ (*,val,t,c,v,*) \in state\_updates \text{ and } view = v \text{ and } ts = t \text{ and } cid = c\}$}
                \STATE{$get\_cview() \equiv cview$}              \COMMENT{this function returns the current view of $s_{i}$}
                \STATE{$get\_dirty\_views() \equiv dirty\_views$} \COMMENT{this function returns the views that $s_{i}$ has participated in}

            \STATEx{\hspace{-1.65em}\textbf{while} $forever$}
                \IF{the timeout for view $s_{i}.cview$ finishes}                                                             \label{alg:line:mod} \COMMENT{the condition that should be satisfied to request for changing the current view}    
                    \STATE{send message $\langle \textbf{change\_view}, cview.succ \rangle$ to other servers, $rchange\_views \leftarrow rchange\_views \cup \{view\}$}     \label{alg:line:send-change-view}
                    \STATE{$schange\_views \leftarrow schange\_views \cup \{cview.succ\}$}                  \COMMENT{the sent request is stored to avoid sending it again}
                \ENDIF
                \IF{$cview.succ \in rchange\_views$}                                                             \label{alg:line:cview.succ-in-change-views} \COMMENT{meeting if there is a request for changing the current view to the succeeding view}
                    \IF{$cview.succ \notin schange\_views$}      \label{alg:line:uninstalling-cview}                                                \COMMENT{if the request for changing the current view is not sent yet, it should be sent}
                        \STATE{send message $\langle \textbf{change\_view}, cview.succ \rangle$ to other servers}
                        \STATE{$schange\_views \leftarrow schange\_views \cup \{cview.succ\}$} \label{alg:line:end-if-schang}             \COMMENT{the sent request is stored to avoid sending it again}
                    \ENDIF
                    \STATE{disable the execution of r/w operations of Algorithm \ref{alg:abd-server}}     \label{alg:line:disable-rw} \COMMENT{r/w operations are disabled to change the view safely}
                    \STATE{$dirty\_views \leftarrow dirty\_views \cup \{cview.succ\}$}                    \label{alg:line:add-dirty-views} \COMMENT{$s_{i}$ is participated in view $cview.succ$ ...}
                    \STATEx{\hfill}                                                                                                                     \COMMENT{... hence, it is required to notify Algorithm \ref{alg:pwr} (Conditions C2R and C2S)}
                    \STATE{$(ts, cid, val) \leftarrow get\_ts\_cid\_val()$ of Algorithm \ref{alg:abd-server}}  \label{line:alg-vc-read-state}       \COMMENT{reading the current state of $s_{i}$'s register}
                    \STATE{send message $\langle \textbf{state\_update}, (val, ts, cid), cview, cview.weight \rangle$ to other servers} \label{alg:line:start-update-state}
                    \STATE{$state\_updates \leftarrow state\_updates \cup \{(s_{i}, val, ts, cid, cview, cview.weight)\}$}
                    \STATE{wait until $n/2 < sum\_weights(cview)$}    \label{alg:line:wait-weighted-majority}                \COMMENT{waiting until a weighted majority of servers with view $cview$ respond}
                    \STATE{$maxts \leftarrow max\_timestamp(cview)$, $maxcid \leftarrow max\_cid(cview,maxts)$}  \label{alg:line:maxts}
                    \STATE{$val \leftarrow val\_maxts\_maxcid(cview,maxts,maxcid)$} \label{alg:line:end-update-state}
                    \STATE{$set\_ts\_cid\_val(maxts, maxcid, val)$ of Algorithm \ref{alg:abd-server}}   \label{alg:line:set-ts-cid-val} \COMMENT{writing the new state of $s_{i}$'s register}
                    \STATE{$cview \leftarrow cview.succ$, set a timeout for $cview$}                                                           \label{alg:line:update-view}    \COMMENT{chaining the view}
                    \STATE{enable the execution of r/w operations}                                               \label{alg:line:enable-rw}
                \ENDIF

            \STATEx{\hspace{-1.65em}\textbf{upon receipt of} message $\langle \textbf{change\_view}, view \rangle$}
                \STATE{$rchange\_views \leftarrow rchange\_views \cup \{view\}$}                              \label{alg:line:change-views}

            \STATEx{\hspace{-1.65em}\textbf{upon receipt of} message $\langle \textbf{state\_update}, (val, ts, cid), view, weight \rangle$ from server $s$}
                \STATE{$state\_updates \leftarrow state\_updates \cup \{(s, val, ts, cid, view, weight)\}$}
        \end{algorithmic}
        \end{algorithm*}

        \begin{example}
            Figure \ref{fig:pwr-views} illustrates an example of executing some pairwise weight reassignments
                and the view changer algorithm.
            In this example, $S = \{s_{1},s_{2},s_{3},s_{4},s_{5}\}$.
            Server $s_{1}$ proposes a pairwise weight reassignment to server $s_{4}$ and another one
                to server $s_{5}$ for view $v_{2}$.
            Also, server $s_{2}$ proposes a pairwise weight reassignment to server $s_{5}$ for view $v_{2}$.
            The proposed pairwise weight reassignments are accepted by servers $s_{4}$ and $s_{5}$.
            At time $t$, the timeout of view $v_{1}$ for server $s_{1}$ finishes.
            Then, server $s_{1}$ sends a message tagged with $\textbf{change\_view}$ to other servers.
            After receiving server $s_{i}$'s message, to change the view, other servers send message $\textbf{change\_view}$ as well.
            Server $s_{3}$ is the first server that changes its view from $v_{1}$ to $v_{2}$ at time $t'$;
                after that, other servers change their views as well.
            The weight of servers $s_{1}$, $s_{4}$, and $s_{5}$ in view $v_{2}$ are $1+2\epsilon$, $1-\epsilon$, and $1-2\epsilon$, respectively.
            Although the pairwise weight reassignment proposed by server $s_{2}$ is accepted, it does not affect $s_{2}$'s weight
                because $s_{2}$ had participated in view $v_{2}$ at the time of receiving the accept message;
                accordingly, the total weight of servers in view $v_{2}$ is $\mathbb{W_{T}} - \epsilon$ until view $v_{2}$ is uninstalled.
            Since server $s_{3}$ has not participated in any pairwise weight reassignment, its weight in view $v_{2}$ is equal to its default weight.
            \begin{figure}
                \centering
                \input{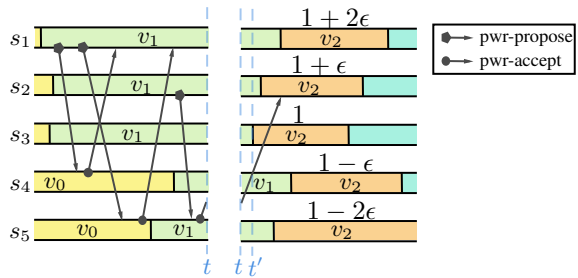}
                \caption{An example of executing some pairwise weight reassignments and the view changer algorithm.
                    The weight of each server in view $v_{2}$ is shown at the top of the view.}
                \label{fig:pwr-views}
                \vspace{-0.3cm}
            \end{figure}
        \end{example}


    \subsection*{Properties of View Changer}\label{subsec:properties}
        We present properties of Algorithm \ref{alg:change-view} in this subsection.
        These properties are used  to prove that:
        (1) for each non-initial view $v$ installed in the system, there exists at least one server that requests to install view $v$, and
        (2) the weight reassignment algorithm satisfies the liveness property.
        Moreover, these properties are used in Section \ref{sec:storage-system} to prove the correctness of the atomic storage.
        Note that all the specified lines in this subsection are related to Algorithm \ref{alg:change-view}.
        \begin{lemma}\label{lem:lem0}
            Let view $v$ be the current view of a server $s$, i.e., $s.cview = v$.
            Server $s$ installs view $v.succ$ if
                at least a weighted majority of servers including $s$ had uninstalled view $v$.
        \end{lemma}
        \begin{proof}
            To install view $v.succ$, server $s$ requires executing Line \ref{alg:line:update-view}.
            To do so, first, it should receive message
                $\langle \textbf{change\_view}, v.succ \rangle$.
            Then, it can execute Lines \ref{alg:line:cview.succ-in-change-views}-\ref{alg:line:start-update-state}.
            If server $s$ passes Line \ref{alg:line:wait-weighted-majority},
                it means that $s$ received messages tagged with \textbf{state\_update}
                from a weighted majority of servers $\rho$.
            Each server $s' \in \rho$ uninstalled view $v$ before sending a message tagged with
                \textbf{state\_update}.
            Therefore, server $s$ can be sure that at least a weighted majority of servers
                uninstalled view $v$ by passing Line \ref{alg:line:wait-weighted-majority}.
            Then, server $s$ can continue executing the remaining lines to change its view to
                $v.succ$.
        \end{proof}

        \begin{lemma}\label{lem:lem01}
            There is only one installed view in the system.
            In other words, if a view $w$ is installed in the system,
                any previously installed view $v < w$ was uninstalled and will not be installed anymore.
        \end{lemma}
        \begin{proof}
            We will argue by contradiction.
            Assume that there are at least two installed views $v$ and $w$ in the system,
                where $v \neq w$.
            Without loss of generality we can assume $v < w$.
            To install $w$ in the system, according to the definition of
                installing a view, at least one server $s$ should install view $w$.
            Due to Lemma \ref{lem:lem0}, at least a weighted majority of servers with view $w.pred$
                including $s$ had uninstalled view $w.pred$ to allow $s$ for installing view $w$.
            Similarly, we can show that at least a weighted majority of servers with view $w.pred.pred$
                including $s$ had uninstalled view $w.pred.pred$.
            With the same argument, we can show that any view $x$ less than $w$ had uninstalled by at least
                a weighted majority of servers with view $x$
                including $s$.
            Therefore, $v$ was uninstalled by at least a weighted majority of servers with view $v$.

            Assume that $v$ was the $i$\textsuperscript{th} installed view in the system and
                $w$ is the $j$\textsuperscript{th} installed view in the system, where $i < j$.
            To install a new view after view $w$, the index of the new view should be greater than $j$.
            Therefore, view $v$ will not be installed anymore.
        \end{proof}

        \begin{lemma}\label{lem:lem1}
            Let $s$ and $s'$ be two correct servers such that:
                (1) $s.cview = v_{j}$,
                (2) $s'.cview = v_{i}$,
                (3) $v_{i} < v_{j}$,
                (4) after uninstalling view $v_{i}$, the sequence of views installed by server $s$ is $\langle v_{i+1}, \dots , v_{j-1}, v_{j} \rangle$, and
                (5) the current view of server $s'$ is less than equal to the current views of
                    all correct servers.
            Server $s'$ installs the same sequence of views $\langle v_{i+1}, \dots , v_{j-1}, v_{j} \rangle$ eventually.
        \end{lemma}
        \begin{proof}
            When $s.cview =  v_{i}$, server $s$ requires executing Line
                \ref{alg:line:cview.succ-in-change-views}-\ref{alg:line:update-view}
                to install view $v_{i+1}$.
            By executing those lines,
                it is guaranteed that two messages are sent by server $s$:
                message $\langle \textbf{change\_view}, v_{i+1} \rangle$
                and message $\langle \textbf{state\_update},*,*,*,v_{i},* \rangle$.
            In addition to server $s'$, there might be some other servers with view $v_{i}$.
            Since server $s$ is correct,
                all correct servers including server $s'$ with view $v_{i}$
                receive message $\langle \textbf{change\_view}, v_{i+1} \rangle$ eventually.
            To finish the proof, we need to show the following cases.
            \begin{itemize}[leftmargin=*]
                \setlength{\itemindent}{0.3em}
                \item[1)] All correct servers including server $s'$ with current view $v_{i}$
                    uninstall view $v_{i}$ and send message $\langle \textbf{state\_update},*,*,*,v_{i},* \rangle$.
                \item[2)] All correct servers including server $s'$ with current view $v_{i}$
                    can only change their views to $v_{i+1}$.
                \item[3)] Server $s'$
                    can pass Line \ref{alg:line:wait-weighted-majority}, i.e.,
                    server $s'$ receives messages
                    $\langle \textbf{state\_update},*,*,*,v_{i},* \rangle$
                    from a weighted majority of servers
                    (then server $s'$ can continue executing Lines \ref{alg:line:maxts}-\ref{alg:line:update-view}
                        to change its view to $v_{i+1}$.)
            \end{itemize}
            The first two cases are straightforward.
            Regarding the third case, according to Assumption \ref{assumption:wl-wu} and Property \ref{property:weight-lower-upper-bound},
                even if $f$ servers fail (each of them with weight $\mathbb{wu}$), there might be a weighted majority of
                correct servers that sent message $\langle \textbf{state\_update},*,*,*,v_{i},* \rangle$.
            Therefore, the last case eventually occurs.
        \end{proof}

        \begin{lemma}\label{lem:lem2}
            Let $lastview = v_{k}$ at time $t$ such that $v_{k}$ is the $k+1$\textsuperscript{th} view installed in the system, where $0 \leq k$.
            There is only one sequence of views $\sigma = \langle v_{0}, v_{1}, \dots, v_{k-1}, v_{k} \rangle$ from $v_{0}$ to $v_{k}$ such that
                $v_{i} = v_{i-1}.succ$ for any $1 \leq i \leq k$.
        \end{lemma}
        \begin{proof}
            Using Lemmas \ref{lem:lem01} and \ref{lem:lem1}, the proof of this lemma is straightforward.
        \end{proof}

        \begin{lemma}\label{lem:lem3}
            Let view $v$ be the last installed view in the system, i.e., $lastview = v$.
            Assuming a correct server $s$ requests to change view $v$ to view $v.succ$,
                then at least a weighted majority of correct servers including $s$ install $v.succ$ eventually.
        \end{lemma}
        \begin{proof}
            Using Lemma \ref{lem:lem1}, all correct servers will eventually reach to view $v$.
            Then, according to Assumption \ref{assumption:wl-wu} and Property \ref{property:weight-lower-upper-bound},
                the total weight of correct servers with view $v$ is a weighted majority.
            If a correct server (say, $s$) with view $v$ sends a request to change its view to view $v.succ$,
                it will change its view due to the existence of a weighted majority of correct servers with view $v$.
            After that, all correct servers will eventually reach to view $v.succ$ by Lemma \ref{lem:lem1}, and
                the total weight of correct servers with view $v$ is a weighted majority according
                to Assumption \ref{assumption:wl-wu} and Property \ref{property:weight-lower-upper-bound}.
        \end{proof}

        \begin{lemma}\label{lem:view-change}
            Let $v$ be the last installed view in the system.
            View $v$ will eventually be changed.
        \end{lemma}
        \begin{proof}
            Due to Lemma \ref{lem:lem3}, there exists (or will eventually exist) at least a
                weighted majority of correct servers with view $v$.
            The timeout of at least one of those servers (say, $s$) for view $v$
                will eventually finish (Line \ref{alg:line:mod} of Algorithm \ref{alg:change-view});
                then $s$ will send message $\langle \textbf{change\_view}, v.succ \rangle$ to other servers.
            Since server $s$ is correct, the view will eventually be changed due to Lemma \ref{lem:lem3}.
        \end{proof}

        \begin{theorem}
            The algorithm is live, i.e.,
                servers can change their weights over time.
        \end{theorem}
        \begin{proof}
            Using Lemmas \ref{lem:lem3} and \ref{lem:view-change}, the proof of this theorem is straightforward.
        \end{proof}

\section{Read/Write Protocols}\label{sec:storage-system}
    In this paper we extend the (static) ABD protocol \cite{abd} to present
        a dynamic weighted atomic storage system that provides atomic r/w protocols.
    Algorithm \ref{alg:abd-client-read} describes r/w protocols executed by clients.
    Besides, Algorithm \ref{alg:abd-server} describes how a server processes r/w operations.
    The main differences between the original ABD protocol and the extended version are as follows.

    \begin{itemize}[leftmargin=*]
        \setlength{\itemindent}{0.3em}
            \item[1)] Each client has a variable $cview$ denoting its current view and initialized with view $v_{0}$ (Line \ref{line:alg-rw:1} of Algorithm \ref{alg:abd-client-read}).
                Each client adds its current view to every r/w request;
                moreover, each client sends its r/w requests to all servers (Line \ref{line:alg-rw:5} of Algorithm \ref{alg:abd-client-read}).
            \item[2)] After receiving each r/w request $r$, each server $s$ determines its
                    current view ($s.cview$) by calling function $get\_cview()$ from Algorithm \ref{alg:change-view},
                    and sets a variable $weight$ to $\perp$
                    (Lines \ref{line:s:getcviewr} and \ref{line:s:getcvieww} of Algorithm \ref{alg:abd-server}).
                For write requests, if the current view of each request is the same as $s.cview$,
                    server $s$ executes the request (Lines \ref{line:s:11}-\ref{line:s:12} of Algorithm \ref{alg:abd-server}).
                Additionally, server $s$ resets the value of variable $weight$ by calling function $get\_weight(s.cview)$ from Algorithm \ref{alg:pwr}
                     (Line \ref{line:s:13} of Algorithm \ref{alg:abd-server}).
                Similarly, for read requests, if the current view of each request is the same as $s.cview$,
                    server $s$ resets the value of variable $weight$ (Line \ref{line:s:7} of Algorithm \ref{alg:abd-server}).
                Then, server $s$ adds $s.cview$ and $weight$ to its response of the client that issued the request.
            \item[3)] Each client updates its current view as soon as it receives a more up-to-date view
                than its current view and restarts the executing operation
                (Lines \ref{line:alg-rw:12}-\ref{alg:line:restart-read} and \ref{line:alg-rw:30}-\ref{alg:line:restart-write} of Algorithm \ref{alg:abd-client-read}).
            \item[4)] Clients consider the weights of servers to decide whether a quorum is constituted
                (Lines \ref{line:alg-rw:15} and \ref{line:alg-rw:33} of Algorithm \ref{alg:abd-client-read}).
    \end{itemize}

    \begin{algorithm}[bt]
        \caption{ABD - client $c_{i}$}
        \label{alg:abd-client-read}
        \small
        \begin{algorithmic}[1]
            \STATEx{\hspace{-1.65em}\textbf{variables}}
                \STATE{$opCnt \leftarrow 0$, $cview \leftarrow v_{0}$}  \label{line:alg-rw:1}

            \STATEx{\hspace{-1.65em}\textbf{functions} to r/w the atomic storage}
                \STATE{$read() \equiv read\_write(\perp)$}
                \STATE{$write(value) \equiv read\_write(value)$}

            \STATEx{\hspace{-1.65em}\textbf{function} $read\_write(value)$}
                \STATEx{\hspace{-0.8em}\textbf{ABD Phase 1}}
                    \STATE{$opCnt \leftarrow opCnt + 1$}
                    \STATE{send $\langle \textbf{read}, opCnt, cview \rangle$ to all servers} \label{line:alg-rw:5}
                    \STATE{$msgs \leftarrow \emptyset$}
                    \REPEAT
                        \STATE{\textbf{upon receipt of} message $\langle \textbf{readack}, val, ts, cid, opCnt, v, w \rangle$}
                        \Indent
                        \vspace{-0.1cm}
                            \IF{$cview = v$}
                                \STATE{$msgs \leftarrow msgs \cup \big\{ (val, ts, cid, opCnt, v, w) \big\}$}
                            \ELSE
                                \IF{$cview < v$}                             \label{line:alg-rw:12}
                                    \STATE{$cview \leftarrow v$}             \label{line:alg-rw:13}
                                \ENDIF
                                \STATE{$read\_write(value)$} \COMMENT{restart the operation} \label{alg:line:restart-read}
                            \ENDIF
                        \EndIndent
                    \UNTIL{$n/2 \leq sum \big(\{w \ | \ (*,*,*,*,*,w) \in msgs\}\big) $}\label{line:alg-rw:15}
                    \IF{$value = \perp$}
                        \STATE{$maxts \leftarrow max\big(\{ts \ | \ (*,ts,*,*,*,*) \in msgs\}\big)$}
                        \STATE{$maxcid \leftarrow max\big(\{cid \ | \ (*,ts,cid,*,*,*) \in msgs$}
                            \STATEx{\hfill $\text{ and } ts = maxts\}\big)$}
                        \STATE{$value \leftarrow \{val \ | \ (val,ts,cid,*,*,*) \in msgs$}
                            \STATEx{\hfill $ \text{ and } ts = maxts \text{ and } cid = maxcid\}$}
                    \ELSE
                        \STATE{$maxts \leftarrow max\big(\{ts \ | \ (*,ts,*,*,*,*) \in msgs\}\big) + 1$}
                        \STATE{$maxcid \leftarrow c_{i}$, $maxval \leftarrow val$}
                    \ENDIF

                \STATEx{\hspace{-0.8em}\textbf{ABD Phase 2}}
                    \STATE{send $\langle \textbf{write}, value, maxts, maxcid, opCnt, cview \rangle$ to all servers}
                    \STATE{$msgs \leftarrow \emptyset$}
                    \REPEAT
                        \STATE{\textbf{upon receipt of} message $\langle \textbf{writeack}, opCnt, v, w \rangle$}
                        \Indent
                        \vspace{-0.1cm}
                            \IF{$cview = v$}
                                \STATE{$msgs \leftarrow msgs \cup \{(opCnt, v, w)\}$}
                            \ELSE
                                \IF{$cview < v$}         \label{line:alg-rw:30}
                                    \STATE{$cview \leftarrow v$}     \label{line:alg-rw:31}
                                \ENDIF
                                \STATE{$read\_write(value)$} \COMMENT{restart the operation} \label{alg:line:restart-write}
                            \ENDIF
                        \EndIndent
                    \UNTIL{$n/2 \leq sum(\{w \ | \ (*,*,w) \in msgs\})$}\label{line:alg-rw:33}
                    \STATE{return $value$}

        \end{algorithmic}
    \end{algorithm}

    \begin{algorithm}[bt]
        \caption{ABD - server $s_{i}$}
        \label{alg:abd-server}
        \small
        \begin{algorithmic}[1]
            \STATEx{\hspace{-1.65em}\textbf{variables}}
                \STATE{$ts \leftarrow 0$, $cid \leftarrow 0$, $val \leftarrow \perp$}

            \STATEx{\hspace{-1.65em}\textbf{functions}}
                \STATE{$get\_ts() \equiv ts$}
                \STATE{$get\_ts\_cid\_val() \equiv (ts, cid, val)$}
                \STATE{$set\_ts\_cid\_val(t, c, v) \equiv ts \leftarrow t; cid \leftarrow c; val \leftarrow v $}

            \STATEx{\hspace{-1.65em}\textbf{upon receipt of} message $\langle \textbf{read}, cnt, v \rangle$ from client $c$}
                \STATE{$cview \leftarrow get\_cview()$ from Algorithm \ref{alg:change-view}, $weight \leftarrow \perp$} \label{line:s:getcviewr}
                \IF{$cview = v$}
                    \STATE{$weight \leftarrow get\_weight(cview)$ from Algorithm \ref{alg:pwr}} \label{line:s:7}
                \ENDIF
                \STATE{send $\langle \textbf{readack}, val, ts, cid, cnt, cview, weight \rangle$ to client $c$}

            \STATEx{\hspace{-1.65em}\textbf{upon receipt of} message $\langle \textbf{write}, val', ts', cid', cnt, v\rangle$ from client $c$}
                \STATE{$cview \leftarrow get\_cview()$ from Algorithm \ref{alg:change-view}, $weight \leftarrow \perp$} \label{line:s:getcvieww}
                \IF{$cview = v$}
                    \IF{$ts' > ts \text{ or } (ts' = ts \text{ and } cid' > cid)$} \label{line:s:11}
                        \STATE{$ts \leftarrow ts'$, $cid \leftarrow cid'$, $val \leftarrow val'$} \label{line:s:12}
                    \ENDIF
                    \STATE{$weight \leftarrow get\_weight(cview)$ from Algorithm \ref{alg:pwr}} \label{line:s:13}
                \ENDIF
                \STATE{send $\langle \textbf{writeack}, cnt, cview, weight \rangle$ to client $c$}
        \end{algorithmic}
    \end{algorithm}

    \subsection{Correctness}
        In the following, we prove our storage system satisfies the properties defined in Section \ref{sec:preliminaries}.

        \vspace{0.5em}
        \noindent\textbf{Storage liveness. \small} We have to prove that
            every r/w operation executed by a correct client eventually terminates.
        To do so, we prove the following theorem.

        \begin{theorem}
            Every r/w operation executed by a correct client eventually terminates.
        \end{theorem}
        \begin{proof}
            Consider a r/w operation $o$ executed by a correct client $c_{i}$.
            For the first phase $ph1$ from $o$,
                client $c_{i}$ sends a message to all servers and waits for servers' responses.
            There are three cases as follows:
            \begin{itemize}[leftmargin=*]
                \setlength{\itemindent}{0.3em}
                \item[1)] a weighted majority of servers that their views are equal to $c_{i}.cview$
                        responded to client $c_{i}$.
                    Consequently, client $c_{i}$ ends $ph1$ and starts the execution of the second phase
                        $ph2$ from $o$ by sending a message to all servers and waits for servers' responses.
                    There are three sub-cases:
                    \begin{itemize}[leftmargin=*]
                        \setlength{\itemindent}{0.5em}
                        \item[1.1)] a weighted majority of servers that their views are equal to $c_{i}.cview$
                                responded to client $c_{i}$.
                            Consequently, client $c_{i}$ ends $ph2$, and the operation terminates.
                        \item[1.2)]
                            Client $c_{i}$ receives a response from a server with view $v$ such that $c_{i}.cview \neq v$.
                            If $c_{i}.cview < v$, client $c_{i}$ updates its current view.
                            Then, the client restarts the operation.
                            We just need to show that this sub-case occurs a finite number of times, i.e.,
                                the operation $o$ is not restarted infinitely many times.
                            We can conclude from Lemma \ref{lem:view-change} that the process of uninstalling view $lastview$ eventually terminates.
                            Besides, according to Assumption \ref{assumption:finite-view-change}, the number of views are finite.
                            Hence, the operation $o$ will eventually be executed.
                        \item[1.3)]
                            Operation $o$ is concurrent with uninstalling view $lastview$ so that
                                r/w operations are disabled.
                            We can conclude from Lemma \ref{lem:view-change} that the process of uninstalling view $lastview$ eventually terminates.
                            After that, r/w operations are enabled again.
                            Then, we reach to other cases or sub-cases.
                    \end{itemize}
                \item[2)]
                    Client $c_{i}$ receives a response from a server with view $v$ such that $c_{i}.cview \neq v$.
                    If $c_{i}.cview < v$, client $c_{i}$ updates its current.
                    Then, the client restarts the operation (cannot start phase $ph2$).
                    This case is similar to sub-case 1.2.
                \item[3)]
                    Operation $o$ is concurrent with uninstalling view $lastview$ so that
                        r/w operations are disabled.
                    This case is similar to sub-case 1.3.
            \end{itemize}
        \end{proof}

        \vspace{0.5em}
        \noindent\textbf{Storage atomicity. \small} We have to prove that the r/w
            protocols of our weighted storage implement an atomic r/w register (Definition \ref{def:atomic-register}).
        The sketch of the proof is as follows.

        \begin{lemma}\label{lem:s1}
            Each phase of a r/w operation always finishes in the last installed view in the system ($lastview$).
        \end{lemma}
        \begin{proof}
            We prove the lemma by contradiction.
            For the sake of contradiction, consider a client $c$ that finishes a phase $ph$ by receiving
                a quorum of replies from servers (say, $\rho$) in view $v < lastview$.
            Without loss of generality, assume that $v.succ = lastview$.
            According to Lemma \ref{lem:lem01}, by installing view $lastview$ in the system,
                at least a weighted majority of servers (say, $\rho'$) in view $v$ had uninstalled
                view $v$; therefore, a weighted majority of servers in view $v$ cannot participate in phase $ph$.
            Since $\rho \cap \rho' = \emptyset$ (quorum intersection property is not satisfied),
                we have a contradiction.
        \end{proof}

        \begin{lemma}\label{lem:s2}
            Each r/w operation always finishes in the last installed view in the system ($lastview$).
        \end{lemma}
        \begin{proof}
            Both read and write phases of each r/w operation of Algorithm \ref{alg:abd-client-read},
                 are executed in the same view.
            Each phase of a r/w operation always finishes in the last installed view in the system
                according to Lemma \ref{lem:s1} that completes the proof.
        \end{proof}

        Let $\mathcal{R}$ be the register and $write(\beta)$ be the operation to write
            the value $\beta$ in $\mathcal{R}$ with an associated timestamp $ts(\beta)$.
        \begin{lemma}\label{lem:asli}
            Let $\alpha$ be the value of last write operation completed in view $v = lastview$.
            Then, in the next view $v.succ$, one of the following cases may happen.
            (1) If there is  no concurrent write operation with changing view from $v$ to $v.succ$,
                then $\alpha$ is propagated to view $v.succ$.
            (2) If changing view from $v$ to $v.succ$ is concurrent with a write operation
                $write(\beta)$ such that $ts(\alpha)<ts(\beta)$, then
                either $\alpha$ or $\beta$ is propagated to $v.succ$.
        \end{lemma}
        \begin{proof}
            \textit{Case 1.} Since there is no concurrent write operation with
                changing view from $v$ to $v.succ$,
                servers by executing Lines \ref{alg:line:start-update-state}-\ref{alg:line:set-ts-cid-val}
                    of Algorithm \ref{alg:change-view},
                propagate $\alpha$ to view $v.succ$.
            \textit{Case 2.} Three following sub-cases should be considered.
                (a) No server in $v$ updated its local state with $\beta, ts(\beta)$ before
                        installing $v.succ$.
                    We can reduce this case to the Case 1.
                (b) Some servers in $v$ store $\beta, ts(\beta)$ and some servers
                        store $\alpha, ts(\alpha)$.
                    Due to Lemmas \ref{lem:s1} and \ref{lem:s2},
                        client restart $write(\beta)$ and tries it later.
                (c) A weighted majority of servers in $v$ stores  $\beta, ts(\beta)$,
                        so this value should be written to view $v.succ$.
        \end{proof}

        \begin{lemma}\label{lem:13}
            Assume that a read operation $read_{1}$ returns a value $\alpha_{1}$
                at time $t_{1}^{e}$, which has an associated timestamp $ts_{1}$.
            A read operation $read_{2}$ started at time $t_{2}^{s}>t_{1}^{e}$ returns a value
                $\alpha_{2}$ associated with a timestamp $ts_{2}$ such that either:
            (1) $ts_{1} = ts_{2}$ and $\alpha_{1} = \alpha_{2}$, or (2) $ts_{1} \leq ts_{2}$ and
                $\alpha_{2}$ was written after $\alpha_{1}$.
        \end{lemma}
        \begin{proof}
            Assume that $read_{1}$ starts at time $t_{1}^{s}$ and $read_{2}$ ends
                at time $t_{2}^{e}$, then $t_{1}^{s} < t_{1}^{e} < t_{2}^{s} < t_{2}^{e}$.
            There are four possible, mutually exclusive, following cases.

            \textit{Case 1.} There is no concurrent changing view.
                For this case, the behavior of the algorithm is the same of the basic protocol.
            \textit{Case 2.} There is a concurrent changing view with $read_{1}$ ($t_{1}^{s}<t' < t_{1}^{e}$).
                Since $t' < t_{1}^{e} < t_{2}^{s}$, the view changer algorithm installs
                    $v.succ$ before both $s_{2}$ starts and $s_{1}$ ends.
                Consequently, from Lemmas \ref{lem:s1} and \ref{lem:s2}
                    both $read_{1}$ and $read_{2}$ finish in view $v.succ$.
                Now, there is no difference between this case and Case 1.
            \textit{Case 3.} There is a concurrent changing view between $read_{1}$ and $read_{2}$
                ($t_{1}^{e}< t' < t_{2}^{s}$).
                By using Lemma \ref{lem:asli}, we can ensure the correctness of this lemma.
            \textit{Case 4.} There is a concurrent changing view with $read_{2}$ ($t_{2}^{s}<t' < t_{2}^{e}$).
                This case is similar to Case 2.
        \end{proof}

        \begin{theorem}
            The r/w protocols implement an atomic r/w register.
        \end{theorem}
        \begin{proof}
            This proof follows directly from Lemma \ref{lem:13}.
        \end{proof}


\section{Performance Evaluation}\label{sec:evaluation}
    In this section, we present a performance evaluation of
        the atomic storage based on our weight reassignment protocol
        to quantify its quorum latency when compared with
		 atomic storage systems based on the following cases:
		(1) the (static) ABD \cite{abd} that uses an SMQS,
		(2) RAMBO \cite{rambo} (a consensus-based reconfiguration protocol), and
		(3) SmartMerge \cite{effModConsensus-freeFST} (a consensus-free reconfiguration protocol).
    We selected SmartMerge because it avoids unacceptable states in contrast to other consensus-free
        reconfiguration protocols (e.g., \cite{effModConsensus-freeFST, dynAtomicStorageWithoutCons,  elasticReconf}).
    We implemented prototypes of the ABD, RAMBO, and SmartMerge protocols in
        the python programming language.
    Besides, we used KOLLAPS \cite{kollaps}, a fully distributed network emulator, to create the network and links' latencies.

    As we explained in Section \ref{sec:introduction}, reconfiguration protocols present two special functions: \textit{join} and \textit{leave}.
    Servers can join/leave the system by calling these functions.
    Reconfiguration protocols require to be adapted to be used as weight requirement protocols.
    To do so, we change \textit{join} and \textit{leave} functions of RAMBO and SmartMerge to \textit{increase} and \textit{decrease} functions, respectively.
    Each server can request to increase/decrease its weight using \textit{increase}/\textit{decrease} functions.
    Particularly, each server
        can call \textit{increase} and \textit{decrease} functions every $\delta$ unit of time ($0 < \delta$ is a constant)
        as follows to increase/decrease its weight.
    Assume that the latency score of server $s$ are  $ls_{t}$ and $ls_{t'}$ respectively at time $t$ and $t' = t + \delta$.
    Also, assume that the total latency scores of servers computed by $s$ are $LS_{t}$ and $LS_{t'}$ respectively at time $t$ and $t'$.
    Server $s$ calls function \textit{increase} (resp. \textit{decrease})
        to increase (resp. decrease) its weight at time $t'$
        if $ls_{t}/LS_{t} + \tau < ls_{t'}/LS_{t'}$ (resp. $ls_{t'}/LS_{t'} + \tau < ls_{t}/LS_{t}$),
        where $\tau$ is a threshold for changing weights.

    We used one 1.8 GHz 64-bit Intel Core i7-8550U, 32GB of RAM machine.
    KOLLAPS executes each server and client in a separate Docker container \cite{docker}, and the containers communicate through the Docker Swarm \cite{dockerSwarm}.
    We set the numbers of servers and clients to five and ten, respectively.
    Moreover, at most, one server can fail ($f=1$).
    Each client sends a new r/w request as soon as receiving the response of the previously sent r/w request.
    Since there is no difference between r/w protocols regarding the number of communication rounds in our r/w protocols, we set the r/w ratio to $0.5$.

    The duration of each run is 200 seconds.
    In each run, latencies of links are changed every $\Delta = 10$ seconds while the processes are unaware of the value of $\Delta$; we set $\epsilon = 0.1$.
    We executed 100 runs and computed the average of the results that is depicted in Figure \ref{fig:evaluation}.
    The average quorum latencies of the ABD, RAMBO, SmartMerge and our protocol are
        $139$, $118$, $128$, and $101$ milliseconds, respectively.
    The ABD protocol requires the responses of three processes to decide whether a quorum is
        constituted while other protocols might constitute their quorums by two processes.
    Therefore, the quorum latencies of other protocols are less than the ABD  on average.
    In the RAMBO protocol, some views might be active at a time, while in our
        protocol, there is only one installed view at any time;
        hence, our protocol outperforms the RAMBO protocol on average.
    In the SmartMerge protocol, servers might pass intermediate views to install a new view;
        besides, for every r/w operation, it is required to communicate with a quorum of processes to be sure that the r/w operation is executed with the most up-to-date weights.
        However, servers with the view equal to $lastview$
        directly change their views to the new view in our protocol.
        Also, each server $s$ knows its weight, i.e., $s$ does not need to communicate with others to determine its most up-to-date weight.
        Hence, the quorum latency of our protocol is less than SmartMerge on average.

    \begin{figure*}[tb]
        \centering
        \includegraphics[scale=0.475]{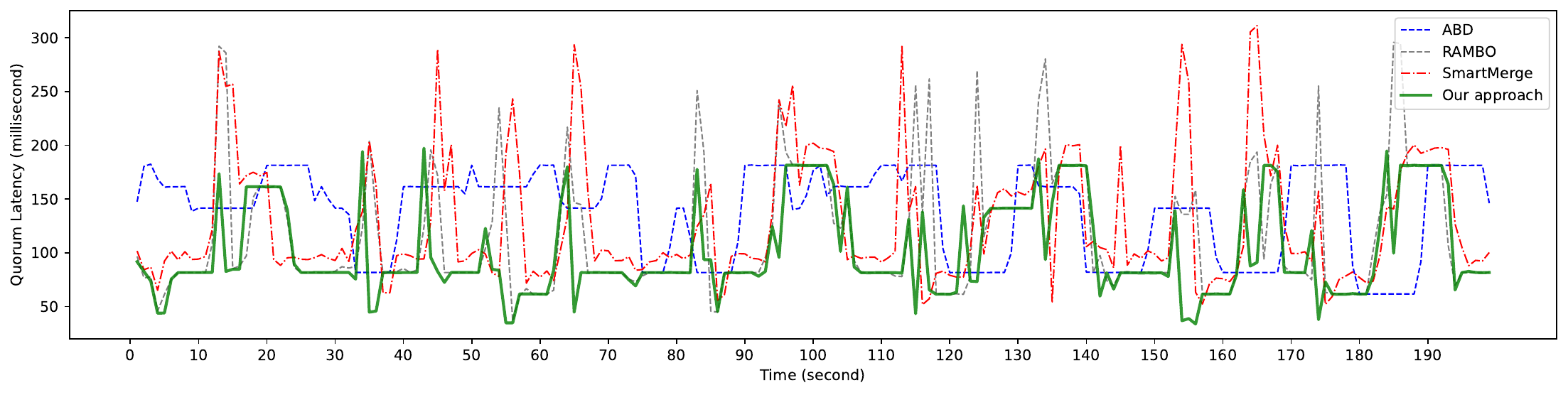}
        \caption{Quorum latency evolution for our protocol, ABD, RAMBO, and SmartMerge.}
        \label{fig:evaluation}
    \end{figure*}

\section{Conclusion and Future Work}\label{sec:conclusion}
	In this paper, we present a novel consensus-free and crash fault-tolerant weight reassignment protocol
		that can be used to improve the performance of atomic read/write storage systems.
	The distinguishing feature of our protocol compared to previous solutions is that for executing each r/w operation,
		it is not required to communicate with a quorum of processes for finding the most up-to-date processes' weights, providing better efficiency.
	The evaluation results show that our protocol outperforms other solutions.
	We assume that the set of servers does not change over time;
		however, the protocol can be extended to consider that servers can
		leave and new servers can join the system as future work.
	Besides, every client sends each of its requests to all servers.
	Working on using strategies for selecting a subset of servers to send
		requests for improving the network congestion can be another direction for future improvement.
	Extending the failure model to Byzantine failures could be another direction for future work as well.

\bibliographystyle{IEEEtran}
\bibliography{../ref.bib}

\appendix\label{appendix-quoracle}
	In Section \ref{sec:introduction}, we presented an example to show
		that throughput of a storage system can be improved using WMQS in contrast to SMQS.
	In Figure \ref{code:quoracle}, we present the code written using Python programming language
		and \textit{quoracle} library to compute throughput of that example.
	We defined four nodes (processes), read quorums for SMQS, and read quorums for WMQS (for simplicity, we omitted the write quorums).
	Finally, the throughput of the storage systems that are based on SMQS and WMQS are computed by calling the \textit{capacity} function.

	\newpage
	\begin{figure}[t!]
		\vspace{-15cm}
		\centering
		\begin{python}
			from quoracle import *
			if __name__ == '__main__':
				p1 = Node('p1', capacity=1000)
				p2 = Node('p2', capacity=800)
				p3 = Node('p3', capacity=400)
				p4 = Node('p4', capacity=200)
				smqs = QuorumSystem(reads=p1*p2*p3 +
																	p1*p2*p4 +
																	p2*p3*p4 +
																	p1*p2*p3*p4)
				wmqs = QuorumSystem(reads=p1*p2 +
																	p1*p3 +
																	p1*p2*p3 +
																	p1*p2*p4 +
																	p2*p3*p4 +
																	p1*p2*p3*p4)
				print(smqs.capacity(read_fraction=1))
				# output: 599.99997
				print(wmsqs.capacity(read_fraction=1))
				# output: 800.0
		\end{python}
		\caption{The code written using \textit{quoracle} library to compare the throughput of SMQS vs. WMQS}
		\label{code:quoracle}
	\end{figure}

\end{document}